\documentclass[stylesheet]{stylesheet}
\bibliography{references.bib}

\begin{document}
\title{\bfseries Pricing of graph-based quantum portfolios from first principles}
\author[1,2,3]{Daniel J. Spencer}
\author[4]{Bin Cheng}
\author[4]{Dinh-Long Vu}
\author[1,5]{Kishor Bharti}
\author[1,2]{Alexey V. Gorshkov}
\author[4,6]{Patrick Rebentrost}

\affil[1]{Joint Center for Quantum Information and Computer Science, NIST/University of Maryland, College Park, MD 20742, USA}
\affil[2]{Joint Quantum Institute, NIST/University of Maryland, College Park, MD 20742, USA}
\affil[3]{Department of Physics, University of Maryland, College Park, MD 20742, USA}
\affil[4]{Centre for Quantum Technologies, National University of Singapore}
\affil[5]{Department of Computer Science and Institute for Advanced Computer Studies, University of Maryland, College Park, Maryland 20742, USA}
\affil[6]{Department of Computer Science, National University of Singapore}

\date{}
\maketitle

\begin{abstract}
    A central topic in economics and finance is the valuation of strategic games and bets, of investments providing uncertain future cashflows, and of other asset-like structures and contracts. Often, such valuation involves a probabilistic model, random variables, and computing expectation values of complex payoffs. Extending beyond this classical framework, a scenario emerges where assets are non-commuting operators and the inherent uncertainty is described by quantum states. In this work, we combine such notions of quantum finance with a graph-theoretic framework and ideas from contextuality. We discuss exclusive quantum assets and introduce scenarios when holding such assets can provide advantages over holding classical assets. In particular, we consider a "quantum market" based on nonlocal games such as the CHSH game and demonstrate that this setting can function as a market with well-defined financial structures. We provide upper bounds on the fair prices of the quantum assets based on graph-theoretic quantities like the independence number and the \lovasz theta number. Our insights into markets of non-commuting assets help to elucidate what a financial system in a world filled with quantum technology could look like.
\end{abstract}

\newpage
\tableofcontents
\newpage

\section{Introduction}\label{sec:intro}
As quantum computers approach the fault-tolerant regime~\cite{shor1996fault-tolerant}, it is important to consider how the computational environment will change with access to such devices. While this is a common discussion in the post-quantum cryptography community~\cite{bernstein2017post-quantum}, other fields able to take advantage of quantum computation should think ahead of time about what changes may arise with the widespread adoption of quantum technologies. One such field is finance, where many quantum algorithms have been proposed~\cite{orus2019quantum,bouland2020prospects,naik2025portfolio,herman2023quantum} for both near-term and fault-tolerant hardware to try to achieve quantum advantage over classical algorithms in routine financial calculations. Common applications of quantum computing in finance include portfolio optimization~\cite{rebentrost2024quantum,kerenidis2019quantum,dalzell2022end-to-end,yalovetzky2023solving,herman2023constrained,lim2024quantum}, derivative pricing~\cite{rebentrost2018quantum, stamatopoulos2020option,chakrabarti2021threshold,an2021quantum-accelerated,doriguello2022quantum,stamatopoulos2024derivative,wang2024option,herman2026quantum}, and risk analysis~\cite{woerner2019quantum,egger2021credit,alcazar2022quantum,leclerc2024financial}.

While these proposed algorithms seek quantum speedups over classical approaches by taking advantage of quantum \emph{computation}, there has been considerably less work on the potential quantum-\emph{information-theoretic} aspects of a financial system in a world in which quantum computers are ubiquitous. In particular, one can conceive of a financial system whose underlying states, assets, and/or payoff mechanisms are intrinsically quantum. There has been some initial work on this, most notably the work by Bao and Rebentrost~\cite{bao2023fundamental} and the work by Hughston and S\'anchez-Betancourt~\cite{hughston2024valuation}. Bao and Rebentrost~\cite{bao2023fundamental} unite ideas from quantum information and finance by introducing a \emph{quantum market} populated with \emph{quantum assets}, and formulate a \emph{Fundamental Theorem of Quantum Asset Pricing} in  analogy to the classical theorem~\cite{follmer2016stochastic,varian1987arbitrage}. Similarly, Hughston and S\'anchez-Betancourt~\cite{hughston2024valuation} take an approach where the payout of a \emph{classical} financial asset is determined by the measurement outcome of a quantum observable. There has also been some recent work on quantum trading~\cite{ding2025coordinating,hymas2026quenching} and game theory~\cite{khan2025quantum}, which have a similar flavor to Refs.~\cite{bao2023fundamental,hughston2024valuation}. We note that these references are distinct from earlier work that tried to explain classical financial markets using ideas from quantum mechanics~\cite{segal1998black-scholes,schaden2002quantum,melnyk2008quantum}.

In this work, we study hypothetical quantum markets through the lens of graph theory, nonlocal games, and quantum contextuality, while also drawing from well-established notions in classical finance. Contextuality is a feature of quantum mechanics that captures the relation between the physical reality of a system and the set of quantum observables used to probe that system. The use of contextuality here is natural because the exclusivity relations between various events translate directly to graph-theoretic constraints, and the resulting graph invariant quantities then determine the attainable values of classical and quantum portfolios. Concretely, we introduce a framework for exclusive quantum assets, demonstrate how and when quantum portfolios can outperform classical ones, and instantiate our framework in a market based on the CHSH game. This provides a market with a clear financial structure and a first-principles pricing mechanism to find fair prices of the exclusive quantum assets that we define. Thus, we not only extend prior definitions and results on quantum assets~\cite{bao2023fundamental,hughston2024valuation}, we also realize these notions in a market framework that yields a realistic price discovery mechanism.

The paper is organized as follows. In the remainder of this introductory section, we give an overview of our main contributions. Then, we provide the relevant background on graph theory and the (non-financial) CHSH game in~\cref{sec:preliminaries}. We describe the CHSH game in further detail and in a framework that we will use to develop our market game in~\cref{sec:chsh-details}. In~\cref{sec:assets}, inspired in part by the work done in~\cite{bao2023fundamental} and~\cite{hughston2024valuation}, we define the classical and quantum assets that will populate our quantum market. Then, in~\cref{sec:chsh-market-game}, we instantiate a quantum market game based on the CHSH game and the assets defined in~\cref{sec:assets}. Finally, we offer some concluding remarks and an outlook in~\cref{sec:discussion}.

\subsection{Main contributions}\label{subsec:main-contributions}
We highlight several novel contributions of our work. In the first part of the paper, we generalize the definitions of quantum assets and quantum Arrow securities, introduced in Refs.~\cite{bao2023fundamental} and~\cite{hughston2024valuation}, respectively. We contrast these with classical assets and classical Arrow securities and examine the properties of all four asset classes. We also introduce the notion of \emph{exclusivity} for these assets, which allows us to use tools from graph theory and contextuality to precisely characterize the assets.

The second part of the paper is a realization of a market that makes use of the quantum assets that we define. In particular, we instantiate such a quantum market via a nonlocal game, namely the CHSH game. While we limit the scope of our work to the CHSH game, we nevertheless provide a framework for other nonlocal games that can be used to realize a quantum market. The use of a specific nonlocal game as the basis of our quantum market also allows us to define specific, concrete quantum assets, so our definitions in the first part of the paper do not remain abstract. In the development of the quantum market based on the CHSH game, we also derive an intuitive pricing mechanism that allows us to discover the price of the quantum assets. Interestingly, this results in the natural emergence of a bid-ask spread, in analogy with the natural spread that arises in the pricing of classical assets in our current financial markets.

\section{Preliminaries}\label{sec:preliminaries}
We begin by introducing some background needed to understand the quantum assets and quantum market game that we develop in this work. We denote by the symbol $\oplus$ the logical \textsc{xor} between two bits and by the symbol $\wedge$ the logical \textsc{and} between two bits. We typically denote bits by lowercase italicized Latin letters such as $a$ and $x$. The Kronecker delta is denoted by $\delta(i,j)$, or, in analogy with random variables, $\delta_i(j)$ and is equal to 1 if $i = j$ and 0 if $i \neq j$. In probability theory~\cite{feller1991introduction}, given a sample space $\Omega$, its elements are called \emph{outcomes} and a subset of $\Omega$ is called an \emph{event}. Sometimes, we simply call $\omega \in \Omega$ an event, which should be understood as the event corresponding to $\{\omega\} \subseteq \Omega$. A random variable is a function $f: \Omega \to \RR$. Given an event $\Omega' \subseteq \Omega$, its indicator function $\mathds{1}_{\Omega'}: \Omega \to \{0,1\}$ is also a random variable, where $\mathds{1}_{\Omega'}(\omega) = 1$ if $\omega \in \Omega'$ and $0$ otherwise.

Given a Hilbert space $\calH$, we denote by $\ket{\psi} \in \calH$ a pure state, which is a complex-valued, $\ell_2$-normalized vector. We denote by $\rho$ a quantum state, where the set of all such mixed states is denoted by $\mathcal{D}(\mathcal{H})$ and consists of all positive semidefinite, Hermitian operators over $\mathcal{H}$ with unit trace. We often call $\rho$ a \emph{density matrix} or a \emph{density operator}. In the rest of this section, we give an overview of the relevant notions from graph theory (\cref{subsec:graph-theory}), followed by a review of the CHSH game (\cref{subsec:chsh-game}) as a motivating example for our quantum market game.

\subsection{Graph theory}\label{subsec:graph-theory}
Here, we review the basic graph-theoretic results relevant for our work. For a detailed exposition, we refer the reader to Refs.~\cite{knuth1993sandwich,grotschel2012geometric}. We denote by $G = (V,E)$ a graph with vertex set $V$ and edge set $E$, where each edge is an unordered pair of distinct vertices. Two vertices $u,v \in V$ are called \emph{adjacent} if $(u,v) \in E$.
The complement of $G$, denoted by $\overline{G} = (V, \overline{E})$, has the same vertex set as $G$ and edge set $\overline{E} = \{(u,v): u,v \in V,\ u \neq v,\ (u,v) \notin E\}$. A \emph{clique} is an induced complete subgraph of $G$. A useful type of graph is a \emph{cycle graph}:

\begin{definition}[Cycle graph]\label{def:cylic-graph}
    A graph with $n$ vertices such that every  vertex $i\in \{0,1, \cdots, n-1\}$ of the graph is connected to the $[(i+1)\pmod{n}]\numth$ and $[(i-1)\pmod{n}]\numth$ vertices is called a \emph{cycle graph} and is denoted by $C_n$.
\end{definition}

If $n$ is odd, then we refer to the cycle graph or subgraph as an \emph{odd cycle}, and analogously an \emph{even cycle} if $n$ is even. The generalization of a cycle graph is a \emph{circulant graph}:

\begin{definition}[Circulant graph]\label{def:circulant-graph}
    Given a set of integers $S \subseteq \{1, \dots, \lfloor n/2 \rfloor\}$, a graph with $n$ vertices where each vertex $i \in \{0,1, \cdots, n-1\}$ is connected to the $[(i+s) \pmod{n}] \numth$ and $[(i-s) \pmod{n}] \numth$ vertices for all $s \in S$ is called a \emph{circulant graph} and is denoted by $\mathrm{Ci}_n(S)$, or $\mathrm{Ci}_n(s_1, \dots, s_k)$, where $S = \{s_1, \dots, s_k\}$.
\end{definition}

\begin{figure}[tb]
    \centering
    \begin{tikzpicture}[
        vertex/.style={
            circle,
            draw,
            fill=white,
            inner sep=1.5pt,
            minimum size=6mm
        }
    ]
    \begin{scope}[xshift=0cm]
        \foreach \i in {0,...,6} {
            \node[vertex] (a\i)
                at ({90-360*\i/7}:2cm) {$\i$};
        }
        \foreach \i in {0,...,6} {
        \pgfmathtruncatemacro{\j}{mod(\i+1,7)}
        \draw (a\i) -- (a\j);
        }

        \node at (0,-2.7) {$\operatorname{Ci}_7(1)$};
    \end{scope}
    
    \begin{scope}[xshift=6cm]
        \foreach \i in {0,...,6} {
            \node[vertex] (b\i)
                at ({90-360*\i/7}:2cm) {$\i$};
        }
    
        \foreach \i in {0,...,6} {
            \pgfmathtruncatemacro{\j}{mod(\i+1,7)}
            \draw (b\i) -- (b\j);
        }
    
        \foreach \i in {0,...,6} {
            \pgfmathtruncatemacro{\j}{mod(\i+2,7)}
            \draw (b\i) -- (b\j);
        }
    
        \node at (0,-2.7) {$\operatorname{Ci}_7(1,2)$};
    \end{scope}
    \end{tikzpicture}
    \caption{Two simple circulant graphs for graphs with $n = 7$ vertices, $\mathrm{Ci}_7(1)$ and $\mathrm{Ci}_7(1,2)$. In $\mathrm{Ci}_7(1)$, each vertex $i$ is connected to vertices $i+1 \pmod{7}$, while in $\mathrm{Ci}_7(1,2)$, each vertex $i$ is connected to both vertices $i+1 \mod{7}$ and $i+2 \pmod{7}$.}
    \label{fig:simple-circulant-graph}
\end{figure}
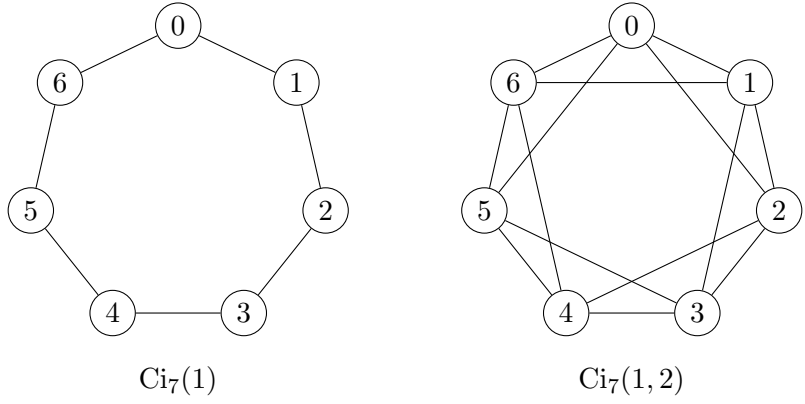

A cycle graph can thus be denoted as $C_n = \mathrm{Ci}_n(1)$. We present two simple circulant graphs, $\mathrm{Ci}_7(1)$ and $\mathrm{Ci}_7(1,2)$, in~\cref{fig:simple-circulant-graph}. Two important graph quantities that we work with throughout this work are the \emph{independence number} and the \emph{\lovasz theta number}. To understand the independence number, we first need the notion of an independent set:

\begin{definition}[Independent set]\label{def:independent-set}
    Let $G = (V,E)$ be a graph. A subset $S \subseteq V$ is called an
    \emph{independent set} if no two distinct vertices in $S$ are adjacent;
    equivalently, for all distinct $i,j \in S$, we have $(i,j) \notin E$.
\end{definition}

If we assign positive weights $\vb{w} = (w_1, \cdots, w_d)^T \in \RR_+^d$ to the vertices of a graph $G$, where $d = \abs{V}$, then we can define the \emph{weighted independence number} and the \emph{weighted \lovasz theta number}~\cite{knuth1993sandwich}. The weighted independence number is defined as:

\begin{definition}[Weighted independence number]\label{def:weighted-independence-number}
    Denote by $\calS$ the set of all independent sets of a graph $G$. Given a positive weight vector $\vb{w} \in \RR_+^d$, the \emph{weighted independence number} of $G$ is defined as
    \begin{equation}\label{eq:ind_number_def_1}
        \alpha(G, \vb{w}) \coloneqq \max_{S \in \calS} \sum_{i \in S} w_i.
    \end{equation}
\end{definition}

The usual (unweighted) \emph{independence number} is recovered by assigning unit weight to every vertex, that is,
\begin{equation}
    \alpha(G) \coloneqq \alpha(G,\vb{1}) = \max_{S \in \calS}\abs{S},
\end{equation}
where $\vb{1}$ denotes the all-ones vector. The independence number is \NP-hard to compute, but it can be upper-bounded by the \lovasz theta number. To define the \lovasz theta number, we first need the definition of an orthonormal representation of a graph:

\begin{definition}[Orthonormal representation]\label{def:orthonormal-representation}
    An \emph{orthonormal representation} of a graph $G = (V, E)$ is a family of unit vectors $\{\ket{v_i}\}_{i \in V}$ in a Hilbert space $\mathcal{H}$ such that for any two distinct nonadjacent vertices $i, j \in V$ (i.e., $(i,j) \notin E$), their corresponding vectors are orthogonal: $\braket{v_i}{v_j} = 0$.
\end{definition}

Then, the weighted \lovasz theta number is defined as follows:

\begin{definition}[Weighted \lovasz theta number]\label{lovasz-theta-number}
    The \emph{\lovasz theta number} $\vartheta(G, \vb{w})$ of a graph $G = (V, E)$ with vertex weights $\vb{w} \in \RR_+^{|V|}$ is defined as
    \begin{align}
        \vartheta(G, \vb{w}) = \max_{\ket{\phi}, \{\ket{v_i}\}}{\defsum{i}{}{w_i \abs{\braket{\phi}{v_i}}^2}},
    \end{align}
    where $\ket{\phi}$ is a unit vector and $\{\ket{v_i}\}$ is an orthonormal representation of a graph complement $\overline{G}$. The vector $\ket{\phi}$ is commonly called the \emph{handle}.
\end{definition}

The usual (unweighted) \emph{\lovasz theta number} is obtained by assigning unit weight to every vertex, that is,
\begin{equation}
    \vartheta(G) \coloneqq \vartheta(G,\vb{1}).
\end{equation}
Equivalently,
\begin{align}
    \vartheta(G) = \max_{\ket{\phi},\{\ket{v_i}\}}\sum_{i\in V} \abs{\braket{\phi}{v_i}}^2.
\end{align}
The \lovasz theta number can be found using semidefinite programming:

\begin{fact}[Finding the \lovasz theta number]\label{fact:finding-lovasz-theta-number}
    For an undirected graph $\ensuremath{G = (V,E)}$ with vertex weights $\vb{w} \in \RR_+^{d}$ and denoting $d = \abs{V}$, the weighted \lovasz theta number can be found by solving the following semidefinite program:
    \begin{equation}\label{theta:primal}
    \begin{aligned} 
        \vartheta(G, \vb{w}) = \max_{X} & \  \sum_{i=1}^d w_i { X}_{ii}\\
        {\rm  s.t.} & \ {X}_{ii} = {X}_{0i},\ \forall i \in [d],\\
        & \ {X}_{ij} = 0,\ \forall (i,j) \in E,\\
        & \ X_{00} = 1,\\
        & \ X \in \mathbb{S}^{d+1}_+,
    \end{aligned}
    \end{equation}
    where $\mathbb{S}^{d+1}_+$ denotes positive semidefinite matrices of size $(d+1) \times (d+1)$ and where $[d] \coloneqq \{1, \ldots, d\}$. 
\end{fact}

Since finding the independence number is an \NP-hard problem, and there exist polynomial-time classical algorithms for semidefinite programs, we can efficiently upper bound the independence number by calculating the \lovasz theta number. For cycle graphs $C_n$ on $n$ vertices, where $n \geq 5$ and $n$ is odd, we have the following exact relations:
\begin{align}
    \alpha(C_n) &= \frac{n-1}{2},\\
    \vartheta(C_n) &= \frac{n\cos{\frac{\pi}{n}}}{1 + \cos{\frac{\pi}{n}}}.
\end{align}

\subsection{CHSH game}\label{subsec:chsh-game}
We now review the CHSH game~\cite{clauser1969proposed}, into which we later instantiate our quantum market. In the CHSH game, Alice and Bob are two players that interact with a neutral referee who administers the game. The referee sends Alice and Bob two random bits $x, y \in \{0,1\}$ and Alice and Bob must respond (without communicating) with bits $a, b \in \{0,1\}$. Alice and Bob win the game if the \emph{predicate} $a \oplus b = x \wedge y$ is satisfied. We are interested in the winning probability when Alice and Bob use a classical strategy and the winning probability when they use a quantum strategy. We will discuss this winning probability in both the \emph{correlation picture} (here) and the \emph{exclusivity graph picture} (in~\cref{sec:chsh-details}).

In the correlation picture, the winning probability is determined by the correlation between Alice and Bob. First, consider the randomness over Alice and Bob's predetermined strategy and denote by $p(a,b \mid x,y)$ the probability that Alice returns $a$ when she is given $x$ from the referee and Bob returns $b$ when he is given $y$ from the referee. Second, consider the randomness of the inputs $x, y$ and denote the probability distribution of the input by $\pi(x,y)$, where $\pi(x,y) = 1/4$ since the referee samples the values of $x$ and $y$ uniformly at random. The probability  of winning the game is given by 
\begin{equation}
    p_\mathrm{win} = \defsum{a \oplus b = x \wedge y}{}{p(a,b \mid x,y) \pi(x,y)}.\label{eq:pwinfirst}
\end{equation}
One can similarly obtain $p_{\mathrm{fail}}$ as the probability of \emph{not} winning the game, where the summation will be over the cases $a \oplus b \neq x \wedge y$. Obviously, $p_{\mathrm{fail}} = 1 - p_{\mathrm{win}}$.

Consider the random variable $(-1)^{a \oplus b}$, which is $1$ if the bits are the same and $-1$ if they are not. The correlation between the output bits given inputs $x,y$ is measured by the conditional expectation value 
\begin{equation}
    E(x, y) \coloneqq \EE[(-1)^{a \oplus b}|x,y] =  \defsum{a,b \in \{0,1\}}{}{p(a,b \mid x,y)(-1)^{a \oplus b}}.\label{eq:condcorr}
\end{equation}
Consider the random variable $(-1)^{x \wedge y}$, for which a high correlation to $(-1)^{a \oplus b}$ indicates that the predicate is more likely to be satisfied. When the predicate is satisfied, $(-1)^{a \oplus b + x \wedge y} = 1$; otherwise, it equals $-1$. The total correlation is given by
\begin{align}
    \frac{S}{4} &\coloneqq \EE[(-1)^{a \oplus b + x \wedge y}]\\
    &= \defsum{x,y}{}{ \pi(x, y) (-1)^{x \wedge y} E(x, y)}\\
    &= \frac{1}{4} (E(0,0) + E(0,1) + E(1,0) - E(1,1)).\label{eq:correlation_CHSH}
\end{align}
The relation between the winning probability and the total correlation is given by  
\begin{equation}
    \frac{S}{4} = p_{\mathrm{win}} - p_{\mathrm{fail}} = 2p_{\mathrm{win}} - 1,
\end{equation}
which follows directly from
\begin{equation}
    \EE[(-1)^{a \oplus b + x \wedge y}] = \defsum{x,y,a,b}{}{p(a,b \mid x,y) \pi(x,y)\left ( \mathbb I_{\{a \oplus b = x \wedge y\}} - \mathbb I_{\{a \oplus b \neq x \wedge y\}} \right )}.
\end{equation}
In summary, the total correlation $S$ is linearly related to the winning probability $p_{\mathrm{win}}$. Later, we will use the exclusivity graph picture to characterize the winning probability, which is in some sense equivalent to the correlation picture associated with the total correlation. 

\paragraph{Classical strategy.}
The set of classical strategies consists of  deterministic and randomized strategies. For example, a simple deterministic strategy is one in which Alice always outputs $a = 0$ and Bob always outputs $b = 0$. The conditional probability $p(a,b \mid x,y)$ in this case is $p(0,0 \mid x,y) = 1$ for all values of $x$ and $y$ and 0 otherwise. In this case, the winning probability is $p_{\mathrm{win}} = 0.75$, corresponding to $S = 2$. 

For more general classical strategies, Alice and Bob can use shared (classical) randomness to determine their outputs, corresponding to the hidden variable model~\cite{einstein1935can,bell2004speakable}. We will use the exclusivity framework to perform the standard analysis for the winning probability of randomized classical strategies below in~\cref{subsec:exclusivity-graph-chsh}. It is well known that using a randomized strategy provides no advantage over an optimal deterministic strategy for the CHSH game, and so the best classical strategy can be thought of as the deterministic one described in the previous paragraph.

\paragraph{Quantum strategy.}
For the quantum strategy, Alice and Bob share an entangled quantum state and agree to measure some local observables according to the inputs $x$ and $y$ that they receive from the referee. Their output bits $a$ and $b$ are determined by the measurement outcomes on these observables. For the optimal quantum strategy, Alice and Bob share a maximally entangled Bell state:
\begin{align}
    \ket{\Phi^+} = \frac{1}{\sqrt{2}}\lp \ket{00} + \ket{11} \rp.
\end{align}
Alice's local observable is taken to be $A_0 = Z$ if she receives $x = 0$ and $A_1 = X$ if she receives $x = 1$. As for Bob's local observables, we allow them to be parameterized: $B_0(\theta) = \cos(2\theta)Z + \sin(2\theta)X$ for $y = 0$ and $B_1(\theta) = \cos(2\theta)Z - \sin(2\theta)X$ for $y = 1$. Note that all these observables have eigenvalues $\pm 1$.

Let $\Pi_{a|x}^A$ be the projector of the $(-1)^a$ eigenstate of $A_x$, and let $\Pi_{b|y}^B$ be the projector of the $(-1)^b$ eigenstate of $B_y(\theta)$. For example, we have $\Pi_{0|0}^A = \ketbra{0}$ and $\Pi_{0|0}^B = \ketbra{b_{0, +}}$, where $\ket{b_{0, +}} = \cos{\theta} \ket{0} + \sin{\theta} \ket{1}$. Then, the conditional probability that determines Alice and Bob's outputs is given by
\begin{equation}
\label{eq:q_condp}
    p(a,b \mid x,y) = \mel{\Phi^+}{\Pi_{a|x}^A \otimes \Pi_{b|y}^B}{\Phi^+}.
\end{equation}
The winning probability is again given by~\cref{eq:pwinfirst}. The different scenarios for the various inputs are given as follows:
\begin{itemize}
    \item When the input is $x = y = 0$, Bob measures $B_0(\theta)$, that is, he measures in the basis $\{\ket{b_{0,+}},\ket{b_{0,-}}\}$. The conditional winning probability is $p(0,0 \mid 0,0) + p(1,1 \mid 0,0) = \abs{\braket{0,b_{0,+}}{\Phi^+}}^2 + \abs{\braket{1,b_{0,-}}{\Phi^+}}^2 = \cos^2(\theta)$.
    \item When the input is $x = y = 1$, Bob measures $B_1(\theta)$, that is, he measures in the basis $\{\ket{b_{1,+}},\ket{b_{1,-}}\}$. The conditional winning probability is $p(0,1 \mid 1,1) + p(1,0 \mid 1,1) = \abs{\braket{+,b_{1,-}}{\Phi^+}}^2 + \abs{\braket{-,b_{1, +}}{\Phi^+}}^2 = \frac{1}{2} + \cos(\theta)\sin(\theta)$.
    \item When the input is $x=1$ and $y=0$, Bob measures $B_0(\theta)$, that is, he measures in the basis $\{\ket{b_{0,+}},\ket{b_{0,-}}\}$. The conditional winning probability is $p(0,0 \mid 1,0) + p(1,1 \mid 1,0) = \abs{\braket{+, b_{0,+}}{\Phi^+}}^2 + \abs{\braket{-,b_{0,-}}{\Phi^+}}^2 = \frac{1}{2} + \cos(\theta)\sin (\theta)$.
    \item When the input is $x=0$ and $y=1$, Bob measures $B_1(\theta)$, that is, he measures in the basis $\{\ket{b_{1,+}},\ket{b_{1,-}}\}$. The conditional winning probability is $p(0,0 \mid 0,1) + p(1,1 \mid 0,1) = \abs{\braket{0, b_{1,+}}{\Phi^+}}^2 + \abs{\braket{1,b_{1,-}}{\Phi^+}}^2 = \cos^2(\theta)$.
\end{itemize}
Combining all of these, the total winning probability as a function of $\theta$ is
\begin{equation}
    p_{\mathrm{win}}(\theta) = \frac{1}{2}\lp \cos^2(\theta) + \cos(\theta)\sin(\theta) + \frac{1}{2} \rp,
\end{equation}
which peaks at $\theta = \pi/8$, yielding $p_{\mathrm{win}}(\pi/8) \approx 0.854$. The parameterized quantum strategy attains the optimal classical winning probability, $p_\mathrm{win} = 3/4$, for the following values of $\theta$:
\begin{equation}
    \theta = \pi n, \pi(n + 1/4),\qquad n \in \ZZ.
\end{equation}
An interval for quantum advantage is thus $\theta \in (0, \pi/4)$, where the boundaries are excluded; see~\cref{fig:success-probability}. In terms of the  correlation, notice the interesting re-expression of the conditional correlation 
\begin{align}
    E(x,y) = \defsum{a,b}{}{\mel{\Phi^+}{\Pi_{a|x}^A \otimes \Pi_{b|y}^B}{\Phi^+} (-1)^{a + b}} = \mel{\Phi^+}{A_x \otimes B_y}{\Phi^+} ,
\end{align}
where we used the fact that $A_x = \Pi_{0|x}^A - \Pi_{1|x}^A$ and $B_y = \Pi_{0|y}^B - \Pi_{1|y}^B$.
When $\theta = \pi/8$, one can show by direct calculation that
\begin{equation}
    E(x,y) = 
    \begin{cases}
        -\frac{1}{\sqrt{2}} & \text{if $x=y=1$}\\
        \frac{1}{\sqrt{2}} & \text{otherwise}.
    \end{cases}
\end{equation}
Thus, for the optimal quantum strategy, $S = 2\sqrt{2}$, corresponding to $p_{\mathrm{win}} \approx 0.854$, which is greater than the optimal classical strategy's winning probability of 0.75. This discussion highlights the quantum-classical separation. The optimality of this value can be derived from Tsirelson’s bound, but the exclusivity graph framework~\cite{cabello2014graph} provides an alternative perspective, which we present below in~\cref{subsec:exclusivity-graph-chsh}.

\begin{figure}[tb]
    \centering
    \includegraphics[width=0.8\textwidth]{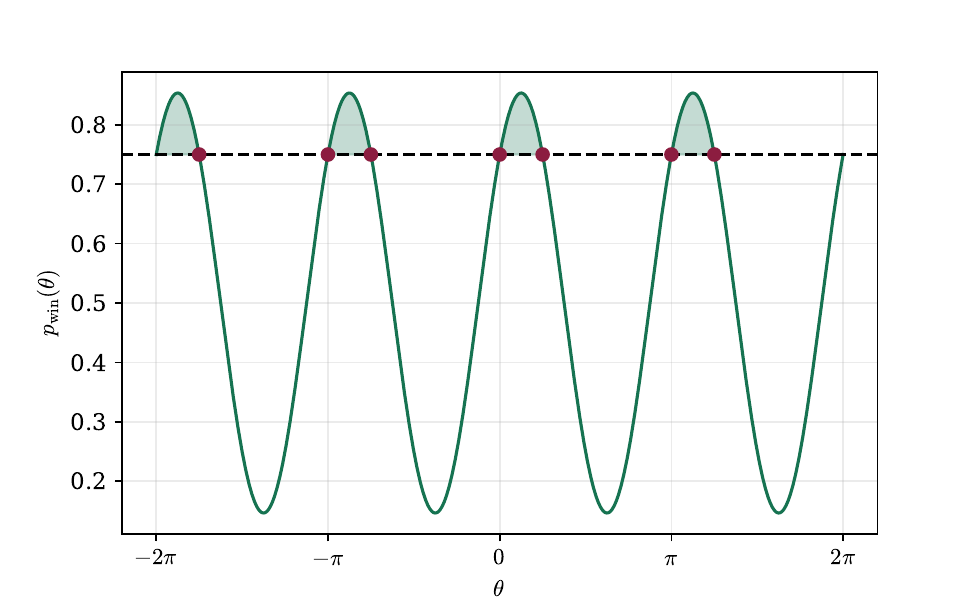}
    \caption{Success probability as a function of $\theta \in (-2\pi, 2\pi)$ for the parameterized projective measurement (i.e., quantum) strategy (green curve) and the classical strategy (black, dashed line), with crossover points indicated in dark red. The areas where the quantum strategy outperforms the classical strategy are shaded in green.}
    \label{fig:success-probability}
\end{figure}

\section{Probabilistic and exclusivity views on CHSH}\label{sec:chsh-details}
In~\cref{subsec:chsh-game}, we reviewed the standard analysis of the CHSH game and the resulting quantum-classical separation. In this section, we study the CHSH game in two alternative ways, a probabilistic view and an exclusivity graph framework. We begin with the probabilistic view.

\subsection{Probabilistic view}\label{subsec:probability-theory-view}
Let us revisit the CHSH game in terms of strategies and random variables. In general, classical strategies are functions that may use some random bits to produce the output bits when given the input bits. When the random bits are fixed, a strategy $f = (f^A, f^B)$ is fully specified by its output for both parties $(f^A_0,f^A_1,f^B_0,f^B_1)$. Hence, a strategy is fully specified by an element of the set $\curly*{0,1}^4$. To model probabilistic strategies, we take a classical strategy (either deterministic or randomized) in the CHSH game to be a probability distribution over the set (sample space)
\begin{equation}
    \Omega = \curly*{\omega : \omega =  (\omega^A_0,\omega^A_1,\omega^B_0,\omega^B_1)\in \curly*{0,1}^4},\label{eq:def_strategyspace}
\end{equation}
where $\omega^A_x$ and $\omega^B_y$ are the outputs of Alice and Bob given input $(x,y) \in \curly*{0,1}^2$, respectively. The optimal deterministic strategy that we discussed before is the Kronecker delta that assigns a probability $1$ to the element $(0,0,0,0)$ of $\Omega$ and $0$ to all other elements. On the other extreme, the completely random strategy corresponds to the uniform distribution over $\Omega$.

For each fixed pair of input-output $\curly*{(x,y),(a,b)}$, we define a random variable $A_{(a,b|x,y)}^\mathrm{cl}: \Omega \to \curly*{0,1}$ as follows. Given a strategy $\omega \in \Omega$, we set $A_{(a,b|x,y)}^\mathrm{cl}(\omega) = 1$ if the input-output pair satisfies the predicate and is compatible with the rule of the strategy, otherwise we set it to $0$. That is, we have 
\begin{align}
    A_{(a,b|x,y)}^\mathrm{cl}(\omega) = \delta(\omega^A_x,a)\delta(\omega^B_y,b)\delta(a\oplus b,x\wedge y).\label{eqn:exclusive-random-variable}
\end{align}
In total, there are 16 of these random variables but since 8 of them are zero by violation of the predicate, we focus our attention only on the 8 that do not violate the predicate. For concreteness, we examine several examples. Consider the input-output pair $(0,0 \mid 0,0)$, where, again, the notation is $(a,b \mid x,y)$ for output $a,b \in \curly*{0,1}$ given input $x,y \in \curly*{0,1}$. There are four strategies in $\Omega$ for which the rule is compatible with this input-output configuration: $(0,0,0,0),(0,0,0,1),(0,1,0,0),(0,1,0,1)$. These are all strategies in which Alice sends $0$ if she receives $0$ and similarly for Bob. In probabilistic language, if one of these events occurs, then the random variable $A^\mathrm{cl}_{(0,0|0,0)}$ is assigned a value of $1$. Equivalently, the random variable $A^\mathrm{cl}_{(0,0|0,0)}$ can be expressed as a sum of four Kronecker delta random variables corresponding to the above events. In general, we can write
 \begin{align}
   \label{eq:Adelta}  A^\mathrm{cl}_{(a,b|x,y)}(\omega') = \defsum{\omega|\omega^A_x=a,\omega^B_y=b}{}{\delta(\omega,\omega')}.
 \end{align}
We again consider the expectation value of these random variables. Given a probability distribution $p$ over $\Omega$ (think of this as a mixed strategy, where with probability $p(\omega)$ we select strategy $\omega$) and a random variable $A^\mathrm{cl}_{(a,b \mid x,y)}$, its expectation is the total probability of all deterministic classical strategies that are compatible with the input-output pair $\curly*{(x,y),(a,b)}$:
\begin{align}
    \EE\bracks*{A^\mathrm{cl}_{(a,b \mid x,y)}} = \defsum{\omega|\omega^A_x=a,\omega^B_y=b}{}{p(\omega)}.
\end{align}
This is exactly the probability that Alice returns $a$ if she receives $x$ and Bob returns $b$ if he receives $y$, that is,
\begin{equation}
    \EE\bracks*{A^\mathrm{cl}_{(a,b \mid x,y)}} = p(a, b \mid x, y).\label{eq:cl_condp}
\end{equation}
Then, substituting this into~\cref{eq:pwinfirst}, the winning probability is given as
\begin{equation}
    p_\mathrm{win} = \defsum{a,b,x,y}{}{\pi(x,y) \EE\bracks*{A^\mathrm{cl}_{(a,b|x,y)}}}.
\end{equation}

\subsection{Exclusivity graph framework}\label{subsec:exclusivity-graph-chsh}
We continue this section by discussing the notion of exclusivity among these random variables in the context of the CHSH game. Intuitively, two events are \emph{exclusive} if they cannot happen simultaneously. For example, in the financial context, two exclusive events could be the price of a bond going either up or down at a certain time; naturally, both events cannot occur simultaneously. Formally, we have the following definition~\cite{dhaene1999safest,feller1991introduction}.

\begin{definition}[Exclusive random variables and events]\label{def:excl-assets}
    Let $\Omega$ be a set. Two random variables $A,B : \Omega \to \mathbb R$  are said to be \emph{exclusive} if $A(\omega) B(\omega) = 0$ for all $\omega \in \Omega$. Two events $E, E' \subseteq \Omega$ are said to be exclusive if $E \cap E' = \emptyset$, which implies that their indicator functions are exclusive random variables, that is, $\mathds{1}_{E} (\omega) \mathds{1}_{E'} (\omega) = 0$ for all $\omega \in \Omega$.
\end{definition}

Note that the inverse, or pre-image, of a random variable returns a subset of $\Omega$~\cite{oksendal2003stochastic}, for example,
\begin{equation}
    \parens*{\mathds{1}_{E}}^{-1}(1) = E\subseteq \Omega.
\end{equation}
Two Kronecker delta random variables $\delta(\omega^1,\omega')$ and $\delta(\omega^2,\omega')$ are exclusive if and only if the corresponding events $\omega^1$ and $\omega^2$ are distinct. Hence, two random variables $A^\mathrm{cl}_{(a,b|x,y)}$ and $A^\mathrm{cl}_{(a',b'|x',y')}$ are exclusive if and only if their supports are disjoint, that is, there exists no deterministic classical strategy $\omega \in \Omega$ compatible with both input-output configurations. In the CHSH game, this occurs when $(x,y)$ and $(x',y')$ overlap in at least one argument, and for that argument, the corresponding outputs are different, that is, $x=x'$ and $a \neq a'$ or $y=y'$ and $b\neq b'$, or both. For example, $A^\mathrm{cl}_{(0,0 \mid 0,0)}$ and $A^\mathrm{cl}_{(1,1 \mid 0,0)}$ are winning random variables and satisfy the predicate $a \oplus b = x \wedge y$, but these random variables are exclusive. Mathematically, the reason is that
\begin{equation}
    \parens*{A^\mathrm{cl}_{(0,0 \mid 0,0)}}^{-1}(1) \cap \parens*{A^\mathrm{cl}_{(1,1 \mid 0,0)}}^{-1}(1) = \emptyset.
\end{equation}
Here, the events corresponding to these pre-images are
\begin{align}
    \parens*{A_{(0,0 \mid 0,0)}^\mathrm{cl}}^{-1}(1) &=: E_{00}\\
    &= \curly*{\omega \in \Omega : \omega_0^A = 0, \omega_0^B = 0}\\
    &= \curly*{(0,0,0,0), (0,0,0,1), (0,1,0,0), (0,1,0,1)}
\end{align}
and
\begin{align}
    \parens*{A_{(1,1 \mid 0,0)}^\mathrm{cl}}^{-1}(1) &=: E_{11}\\
    &= \curly*{\omega \in \Omega : \omega_0^A = 1, \omega_0^B = 1}\\
    &= \curly*{(1,0,1,0), (1,0,1,1), (1,1,1,0), (1,1,1,1)}.
\end{align}
Graph theory provides a tool to systematically account for exclusivity.

The exclusivity graph framework also provides a method for computing the upper bound on the maximal classical and quantum winning probabilities~\cite{cabello2010noncontextuality,cabello2013basic,cabello2014graph,bharti2022graph-theoretic} in the CHSH game and other nonlocal games. First, we recall the random variables defined in~\cref{eqn:exclusive-random-variable} and we think of $(a,b \mid x,y)$ as associated subsets of $\Omega$, which we now call \emph{conditional events} in the CHSH game. Each event occurs with probability $p(a,b \mid x,y)$; see~\cref{subsec:chsh-game} and the winning probability defined therein. Recall also that there are 16 such events corresponding to the 16 different eigenstates of the operators $A_x \otimes B_y$; again, see~\cref{subsec:chsh-game}. Among these events, 8 of them lead to winning the CHSH game, corresponding to the events that have eigenvalue $(-1)^{x \wedge y}$. The winning events and their corresponding eigenstates are listed in~\cref{fig:chsh-excl-graph}. 

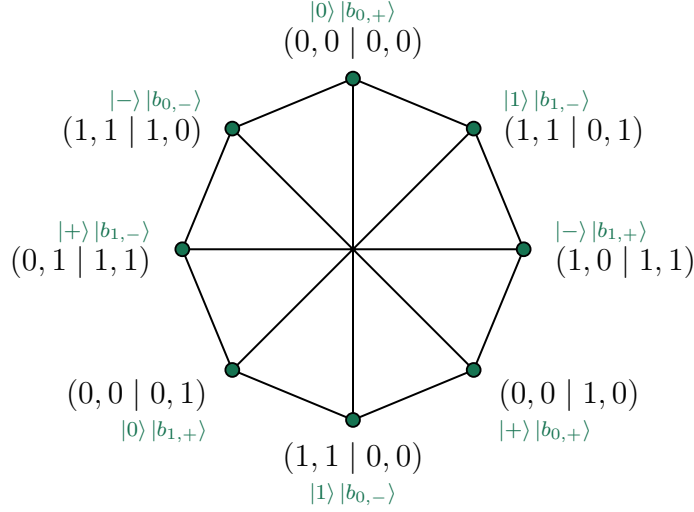
\begin{figure}[tb]
    \centering
    \begin{tikzpicture}[
        scale=1.05,
        edge/.style={black, line width=0.75pt, line cap=round},
        vertex/.style={
            circle,
            draw=black,
            fill=themecolor,
            minimum size=5.2pt,
            inner sep=0pt,
            line width=0.7pt
        },
        lab/.style={
            inner sep=1pt,
            outer sep=0pt
        }
    ]
    \newcommand{\state}[1]{\textcolor{themecolor}{\scriptsize $#1$}}
    \newcommand{\event}[1]{{\large $#1$}}
    \newcommand{\StateOverEvent}[2]{\state{#1}\\[-1pt]\event{#2}}
    \newcommand{\EventOverState}[2]{\event{#2}\\[-1pt]\state{#1}}
    
    \def\R{2.15}
    \foreach \i/\ang in {
        0/90,
        1/45,
        2/0,
        3/-45,
        4/-90,
        5/-135,
        6/180,
        7/135
    }{\coordinate (v\i) at (\ang:\R);}
    
    \draw[edge] (v0)--(v1)--(v2)--(v3)--(v4)--(v5)--(v6)--(v7)--cycle;
    \draw[edge]
        (v0)--(v4)
        (v1)--(v5)
        (v2)--(v6)
        (v3)--(v7);
    
    \foreach \i in {0,...,7}{
        \node[vertex] at (v\i) {};
    }
    
    \node[lab, anchor=south, align=center] at ($(v0)+(0,0.22)$)
        {\StateOverEvent
            {\ket{0}\ket{b_{0,+}}}
            {(0,0 \mid 0,0)}};
    
    \node[lab, anchor=west, align=left] at ($(v1)+(0.34,0.12)$)
        {\StateOverEvent
            {\ket{1}\ket{b_{1,-}}}
            {(1,1 \mid 0,1)}};
    
    \node[lab, anchor=west, align=left] at ($(v2)+(0.36,0)$)
        {\StateOverEvent
            {\ket{-}\ket{b_{1,+}}}
            {(1,0 \mid 1,1)}};
    
    \node[lab, anchor=north west, align=left] at ($(v3)+(0.28,-0.08)$)
        {\EventOverState
            {\ket{+}\ket{b_{0,+}}}
            {(0,0 \mid 1,0)}};
    
    \node[lab, anchor=north, align=center] at ($(v4)+(0,-0.24)$)
        {\EventOverState
            {\ket{1}\ket{b_{0,-}}}
            {(1,1 \mid 0,0)}};
    
    \node[lab, anchor=north east, align=right] at ($(v5)+(-0.28,-0.08)$)
        {\EventOverState
            {\ket{0}\ket{b_{1,+}}}
            {(0,0 \mid 0,1)}};
    
    \node[lab, anchor=east, align=right] at ($(v6)+(-0.36,0)$)
        {\StateOverEvent
            {\ket{+}\ket{b_{1,-}}}
            {(0,1 \mid 1,1)}};
    
    \node[lab, anchor=east, align=right] at ($(v7)+(-0.34,0.12)$)
        {\StateOverEvent
            {\ket{-}\ket{b_{0,-}}}
            {(1,1 \mid 1,0)}};
    \end{tikzpicture}
    \caption{Exclusivity graph corresponding to the eight winning events in the CHSH scenario. This graph is a subgraph of the 16-vertex graph of all events. This graph is called the eight-vertex circulant graph $(1,4)$ and is denoted $\mathrm{Ci}_8(1,4)$. Each vertex is an event $(a,b \mid x,y)$, where $a,b \in \{0,1\}$ are the output bits given input bits $x,y \in \{0,1\}$, and an edge connects events assigning different outcomes to the same local measurement. Explicitly, the states are $\ket{b_0,+} = \cos{\theta}\ket{0} + \sin{\theta}\ket{1}$, $\ket{b_0,-} = -\sin{\theta}\ket{0} + \cos{\theta}\ket{1}$, $\ket{b_1,+} = \cos{\theta}\ket{0} - \sin{\theta}\ket{1}$, and $\ket{b_1,-} = \sin{\theta}\ket{0} + \cos{\theta}\ket{1}$.}
    \label{fig:chsh-excl-graph}
\end{figure}

From here, we can construct a graph, called an \emph{exclusivity graph}, to represent all of these events (the vertices) and their exclusivity relation (edges); see~\cref{fig:chsh-excl-graph}. This graph is an eight-vertex circulant graph $\mathrm{Ci}_8(1,4)$, which is a well-studied graph in graph theory. Let
\begin{equation}
    \Sigma \coloneqq \defsum{a \oplus b = x \wedge y}{}{p(a,b \mid x,y)},
\end{equation}
such that $p_{\mathrm{win}} = \Sigma/4$. Here, the input pair $(x,y)$ specifies a \emph{measurement context}. A \emph{non-contextual model} assigns outputs to all local measurements in advance, with Alice's assigned output for input $x$ being independent of $y$, and Bob's assigned output for input $y$ being independent of $x$. Probabilistic non-contextual models are mixtures of these assignments, as in the classical strategies that we saw above in~\cref{subsec:probability-theory-view}. Correlations that admit no such description are called \emph{contextual}. In each deterministic assignment, the winning events assigned value 1 form an independent set of the exclusivity graph. As such, both deterministic assignments and their probabilistic mixtures satisfy the upper bound on $\Sigma$ given by the \emph{independence number} of the exclusivity graph~\cite{cabello2014graph}, which in this case is $\alpha(\mathrm{Ci}_8(1,4)) = 3$, so we have
\begin{equation}
    \Sigma \leq \alpha(\mathrm{Ci}_8(1,4)) = 3.
\end{equation}
For the CHSH game, this bound is tight. Quantum correlations, on the other hand, need not admit such a non-contextual description. In a quantum theory, the upper bound is given by the \lovasz theta number~\cite{cabello2014graph}, which in this case is $\vartheta(\mathrm{Ci}_8(1,4)) = 2 + \sqrt{2} \approx 3.414$, so we have
\begin{equation}
    \Sigma \leq \vartheta(\mathrm{Ci}_8(1,4)) \approx 3.414.
\end{equation}
For the CHSH game, this is also tight. Therefore, for the CHSH game, the optimal winning probabilities are
\begin{equation}
    p_{\mathrm{win,classical}} = \frac{3}{4}\quad\text{and}\quad p_{\mathrm{win,quantum}} = \frac{2+\sqrt{2}}{4} \approx 0.854.
\end{equation}
This gives us the well-known quantum-classical separation in the CHSH game, but here derived from graph-theoretic quantities.

\section{Classical and quantum assets}\label{sec:assets}
We now develop the framework for a \emph{quantum market} that allows us to price certain classes of portfolios from first principles. We do so by first introducing classical assets and how to incorporate them into the CHSH game in~\cref{subsec:classical-assets,subsec:classical-assets-chsh}. Then, we introduce the analogous quantum assets and their role in the CHSH game in~\cref{subsec:quantum-assets,subsec:quantum-assets-chsh}. We then tie together the formalism of classical and quantum assets with the notion of exclusivity that we introduced in the previous section by introducing exclusive assets in~\cref{subsec:exclusive-assets} and exclusive asset portfolios in~\cref{subsec:exclusive-asset-portfolios}. We provide a summary of the contrasts between the classical and quantum models in~\cref{tab:models} for reference.

\begin{table}[tb]
    \centering
    \begin{tblr}{colspec={X[c,m]X[c,m]X[c,m]},hlines,vlines,row{even}={bg=themecolor!50},row{1}={font=\bfseries}}
        Classical model & Quantum model & Financial interpretation\\
        Sample space $\Omega$ & Hilbert space $\mathcal{H}$ & Possible underlying states\\
        Probability vector $\vb{p}$ & Density operator $\rho$ & Probabilistic market description\\
        Event indicator $\mathds{1}_E$ & Projector $P$ & Binary contingent claim\\
        $\defsum{\omega}{}{p(\omega)A(\omega)}$ & $\Tr[\rho A]$ & Expected payoff\\
        Disjoint supports & Orthogonal supports & Payoff exclusivity
    \end{tblr}
    \caption{Summary of different market models and their financial interpretation.}
    \label{tab:models}
\end{table}

\subsection{Classical assets}\label{subsec:classical-assets}
Starting from the classical perspective, we define an \emph{asset} as any object that has value at some future time point. Here, by "having value," we mean that the asset allows for the receipt of, or the exchange of, some numeraire currency such as the United States dollar (USD) with nonzero probability at some future point in time. Thus, in addition to well-known assets like houses, assets include contracts with the promise of future cashflows such as bonds, equities and their dividends, commodities futures, and bets or gambles. In particular, the mere \emph{possibility} for the owner of an asset to receive a future cashflow inherently gives value to the asset in the present moment. Often, there is a \emph{market mechanism} that determines the \emph{price} corresponding to the market's assessment of the fair value of the asset. For example, for highly liquid S\&P 500 constituent stocks, the prices are determined by the frequent buying and selling on public exchanges like the New York Stock Exchange and the NASDAQ. In more theoretical cases such as when a market model is given, the price itself or the price discovery mechanism can be described from a first-principles analysis. 

We consider a single-period model~\cite{follmer2016stochastic}, defined by a present time $t=0$ and a future time $t = t_\mathrm{final}$. At the present time, a certain number of financial assets are defined and buy and sell orders are executed. The outcomes at the future time are described probabilistically. The single-period setting is used for simplicity, but corresponds to real-world situations such as a single dividend, a lottery win, or the exercise of \emph{European-style} options, which are exercised (i.e., cashed out) only once at the maturity time of the option\footnote{This is in contrast to, for example, \emph{American-style} options, where the buyer can exercise the asset at any time between the purchase time $t=0$ and the maturity time $t=t_\mathrm{final}$.}. In these scenarios, an investor buying the asset at time $t=0$ will receive a cashflow at the maturity time $t=t_\mathrm{final}$, and the cashflow depends on probabilistic outcomes at $t=t_\mathrm{final}$.

First, we describe the ``classical'' probabilistic scenario~\cite{follmer2016stochastic}. Let the set of financial outcomes or events at time $t_\mathrm{final}$ be described by the sample space $\Omega$. For simplicity, we consider the case that $\Omega$ is finite and we denote $K = \abs{\Omega}$. The probability that each event $\omega_i \in \Omega$ occurs at the future time is described by $\vb{p} \in \Delta^K$, where $\Delta^K$ is the simplex consisting of non-negative, $\ell_1$-normalized $K$-dimensional vectors. More explicitly, we take the sample space to be $\Omega = \{\omega_1, \ldots, \omega_K\}$ and the probability distribution to be $\vb{p} = (p_1, \ldots, p_K)$, where each event $\omega_i$ has probability $p_i$ of occurring for $i \in [K]$ and where $\defsum{i=1}{K}{p_i} = 1$. Then, we define a classical asset as follows:

\begin{definition}[Classical asset]\label{def:classical-asset}
    A \emph{classical asset} $A: \Omega \to \mathcal{E}$ is a random variable assigning each event $\omega_i \in \Omega$ to a \emph{cashflow} or \emph{payout structure} $\mathcal{E}$ at time $t_\mathrm{final}$, where $\mathcal{E}$ could be binary values $\{0,1\}$, the real numbers $\RR$, or any other set.
\end{definition}

Consider a basic market where all assets have a real-valued outcome determined by the elements from $\Omega$. We express the assets using linear combinations of \emph{Arrow securities}~\cite{arrow1954existence}, or \emph{Arrow-Debreu securities}, which are theoretical securities that pay one unit of a numeraire (e.g., a dollar) when some event $\omega \in \Omega$ happens, and nothing otherwise. Formally, Arrow securities can be defined as follows:

\begin{definition}[Classical Arrow security]\label{def:classical-arrow-security}
    Let $\Omega$ be a sample space. A \emph{classical Arrow security} is a random variable $A^{(\omega)} : \Omega \to \{0,1\}$ whose payout structure is given by
    \begin{equation}
        A^{(\omega)}(\omega') =
        \begin{cases}
            1 & \text{if $\omega' = \omega$}\\
            0 & \text{otherwise,}
        \end{cases}
    \end{equation}
    where $\omega, \omega' \in \Omega$.
\end{definition}

In many financial models, it is assumed that there exists a classical asset that achieves a fixed rate of return. That is, the \emph{risk-free asset} is an investment with a rate of return $r > 0$ for all $\omega \in \Omega$. The risk-free rate is used to discount future cashflows to today's value. Intuitively, the risk-free rate is an alternative investment that a trader could make that is guaranteed to pay them $(1+r)$ times their initial investment. This, however, is at the cost of missing out on potentially higher returns from a riskier investment, namely the purchase of one of the other assets with nonzero risk. In real-world scenarios, the risk-free rate is often taken to be the rate of return of United States Treasury bills (T-bills).

\subsection{Event-based classical assets in the CHSH setting}\label{subsec:classical-assets-chsh}
Recall the sample space $\Omega$ in~\cref{eq:def_strategyspace} and the random variables in~\cref{eqn:exclusive-random-variable}. Let $\vb{p}$ be a probability measure on the subsets of $\Omega$. We can think of $A^\mathrm{cl}_{(a,b \vert x,y)}$ as assets under~\cref{def:classical-asset} if the following holds. First, the buyer of an asset typically is assumed to have no control over the probability measure $\vb{p}$, where $\vb{p}$ describes the state of the world at some future time. In the CHSH setting, Alice and Bob can decide beforehand on a strategy, that is, the probability measure, with the optimal classical strategy simply being $p(\omega) = \delta_{\omega,(0,0,0,0)}$. However, to align with the financial setting, we consider that the holder of the random variables $A^\mathrm{cl}_{(a,b|x,y)}$ does not precisely know or control the strategy embodied in $\vb{p}$. Second, we consider the outcome of the random variable to be associated with a cashflow. Specifically, we take $A^\mathrm{cl}_{(a,b|x,y)} = 1$ to mean a payoff of $\$1$. Note that, for example,
\begin{equation}
    A^\mathrm{cl}_{(0,0|0,0)}(\omega) = \delta_{\omega,(0,0,0,0)} +  \delta_{\omega,(0,0,0,1)} +\delta_{\omega,(0,1,0,0)} + \delta_{\omega,(0,1,0,1)},
\end{equation}
highlighting that the asset is a linear combination of four Arrow securities. The same holds true for the other $A^\mathrm{cl}_{(a,b|x,y)}$. Finally, there is some market mechanism that allows for the exchange of these assets. In~\cref{sec:chsh-market-game}, we provide a discussion of such a market. Under these conditions, then, we consider $A^\mathrm{cl}_{(a,b|x,y)}$ to be classical assets that can be used in a market setting instantiated via the playing of a CHSH game. We make this more explicit as we develop the quantum analogues in the rest of this section and in~\cref{sec:chsh-market-game}.

\subsection{Quantum assets}\label{subsec:quantum-assets}
We now expand the above to \emph{quantum markets} and \emph {quantum assets}. The works of Bao and Rebentrost~\cite{bao2023fundamental} and Hughston and S\'{a}nchez-Betancourt~\cite{hughston2024valuation} argue that there could arise scenarios in which the probabilistic description of the future time $t_\mathrm{final}$ involves a quantum state in contrast to the classical probability vector $\vb{p}$. For this scenario, we consider a Hilbert space $\mathcal{H}$ in contrast to the classical sample space $\Omega$, and a setting where the future state is described by a density operator $\rho \in \mathcal{D}(\mathcal{H})$.  We call this density operator the "market density operator" or simply the "market state." Formally, we have:

\begin{definition}[Single-period quantum market density operator]\label{def:market-density-operator}
    Let $\mathcal{H}$ be a Hilbert space and $\mathcal{D}(\mathcal{H})$ be the set of density operators defined on $\mathcal{H}$. A single-period quantum market is a market where the probabilistic nature of the market at some future time $t=t_\mathrm{final}$ is described by a \emph{market density operator} $\rho \in \mathcal{D}(\mathcal{H})$.
\end{definition}

For any $\rho \in \mathcal{D}(\mathcal{H})$, there exists a unitary operator such that $\rho = U \mathrm{diag}(\vb{p}) U^\dagger$ for some probability distribution $\vb{p}$. The diagonal density operator corresponding to the classical market is $\rho_\mathrm{cl} = {\rm diag}(\vb{p})$, where $\vb{p}$ is a probability distribution in the fixed reference basis corresponding to classical events.

In this way, we can see that quantum assets are a generalization of the classical assets that we defined in~\cref{subsec:classical-assets}. A quantum asset remains an asset that provides a cashflow to the owner at the future time, or, more generally, such a quantum asset has an influence on the amount of this cashflow. The works of Bao and Rebentrost~\cite{bao2023fundamental} and Hughston and S\'{a}nchez-Betancourt~\cite{hughston2024valuation} argue that such assets arise in scenarios described by a market density operator. In both~\cite{bao2023fundamental} and~\cite{hughston2024valuation}, a quantum asset or a quantum contract is defined as a Hermitian operator. One possible scenario is that an investor buying the quantum asset at time $t=0$ will receive a cashflow at the maturity time $t=t_\mathrm{final}$ depending on the measurement outcome of that asset, where the cashflow amount is encoded in the eigenvalues of the quantum asset. Similar to~\cite{bao2023fundamental} and~\cite{hughston2024valuation}, for our work, we define a quantum asset as follows:

\begin{definition}[Quantum asset]\label{def:quantum-asset}
    Let $\mathcal{H}$ be a $K$-dimensional Hilbert space. A \emph{quantum asset} is a Hermitian operator $A \in \CC^{K \times K}$. If we take the spectral decomposition of $A$ to be
    \begin{equation}
        A = \sum_{k=1}^K a_{k} P_{k},
    \end{equation}
    then holding one unit of $A$ gives a payoff of $a_{k}$ units of the numeraire when the projective measurement $\{P_{k}\}_{k \in [K]}$ at time $t=t_\mathrm{final}$ returns outcome $k$. For a market state $\rho \in \calD(\mathcal{H})$, the expected payoff at time $t = t_\mathrm{final}$ is
    \begin{equation}
        \Tr[\rho A] = \sum_{k=1}^K a_{k}\Tr[\rho P_{k}].
    \end{equation}
\end{definition}

We can also consider a natural quantum analog of the classical Arrow securities that we defined in~\cref{def:classical-arrow-security}, where a quantum Arrow security is a rank-1 projector~\cite{hughston2024valuation}. Formally, we have the following:

\begin{definition}[Projector quantum Arrow security]\label{def:quantum-arrow-security}
    Let $\mathcal{H}$ be a $K$-dimensional Hilbert space. For each event $k \in [K]$, the \emph{quantum Arrow security} associated with event $k$ is a rank-1 projector $P_k$, where $P_k P_{k'} = \delta_{kk'}P_k$ and $\sum_{k=1}^{K}{P_k} = I$, which pays one unit of a numeraire if the projective measurement $\curly*{P_k, I-P_k}$ at maturity returns the outcome associated with $P_k$ and zero otherwise.
\end{definition}

Often, classical assets are represented by linear combinations of classical Arrow securities. In analogy, one can attempt to express quantum assets in terms of quantum Arrow securities. For example, a Hermitian asset $A$ with eigenstates $\{\ket{k}\}$ can be expanded as
\begin{equation}
    A = \sum_{k=1}^K a_{k} \ketbra{k},
\end{equation}
where $\{a_{k}\}$ are the corresponding eigenvalues. Of course, several non-commuting assets cannot all be expressed as linear combinations of a single complete family of mutually orthogonal quantum Arrow securities, as this would imply that they are simultaneously diagonalizable. In~\cref{subsec:quantum-assets-chsh}, we will define appropriate quantum Arrow securities for the CHSH setting. Given multiple assets, we can construct a \emph{portfolio} of quantum assets:

\begin{definition}[Portfolio]\label{def:portfolio}
    Let $d \in \NN$ and let $\vb{A} = \lp A_1, A_2, \ldots, A_{d} \rp \in \CC^{d \times K \times K}$ be a \emph{basket} of quantum assets. A \emph{portfolio} $\vb{w} \in \ZZ_{\geq 0}^{d}$ is a vector encoding the amount of each asset an investor holds.
\end{definition}

In this work, we restrict shares to be non-negative integers, though future work could consider fractional shares or negative shares, which would allow for shorting. The \emph{expected value} of a portfolio $\vb{w}$ is then defined as follows:

\begin{definition}[Expected portfolio value]\label{def:expected-portfolio-value}
    In the setting of~\cref{def:market-density-operator,def:quantum-asset,def:portfolio}, the \emph{expected portfolio value} in $t = t_\mathrm{final}$ currency units is given by
    \begin{align}
        \EE_\rho[\vb{w}] \coloneqq \Tr[\rho\lp \vb{w} \cdot \vb{A} \rp] = \sum_{i=1}^d w_i \Tr[\rho A_i].
    \end{align}
\end{definition}

This expected portfolio value will allow us to determine fair bid and ask prices in the quantum market game that we develop in~\cref{sec:chsh-market-game}.

\subsection{Event-based quantum assets in the CHSH setting}\label{subsec:quantum-assets-chsh}
We now apply these definitions to the CHSH setting, where the projectors associated with input–output events provide concrete examples of quantum assets and quantum Arrow securities. In the quantum case, a CHSH strategy consists of a shared quantum state $\rho$ and local projective measurements $\{\Pi^A_{a|x}\}_{a=0,1}$ and $\{\Pi^B_{b|y}\}_{b=0,1}$ for each input $x,y \in \{0,1\}$. Since these measurements have two outcomes, their projectors satisfy
\begin{equation}
    \Pi^A_{1|x} = I_A-\Pi^A_{0|x},\qquad \Pi^B_{1|y} = I_B-\Pi^B_{0|y}.
\end{equation}
For each input–output configuration $(a,b \mid x,y)$, the corresponding event projector is
\begin{equation}
    \Pi_{(a,b|x,y)} = \Pi^A_{a|x} \otimes \Pi^B_{b|y},
\end{equation}
and its conditional probability is
\begin{equation}
    p(a,b \mid x,y) = \Tr[\rho\,\Pi_{(a,b|x,y)}].
\end{equation}
For winning configurations, this parallels the classical identity $\EE_p[A^\mathrm{cl}_{(a,b|x,y)}] = p(a,b \mid x,y)$ from~\cref{eq:cl_condp}: the classical indicator and the quantum projector represent the corresponding event in their respective models.

The projectors $\Pi_{a|x}^A \otimes \Pi_{b|y}^B$ can be considered quantum assets, according to~\cref{def:quantum-asset}, when the following criteria are met. First, we consider the state $\rho$ as the world/market state. The holder of the asset has no control over the preparation of this state, in contrast to the standard CHSH game where Alice and Bob prepare the state as part of their strategy. Second, the holder of the asset obtains \$1 if the local projective measurements of $\{\Pi^A_{a'|x}\}_{a'=0,1}$ and $\{\Pi^B_{b'|y}\}_{b'=0,1}$ on the market state return outcomes labeled $a$ and $b$, respectively, and receives zero otherwise. Finally, there exists a market that allows the trading of claims contingent on these measurement outcomes; we make all of this more concrete in~\cref{sec:chsh-market-game} when we develop a CHSH game-based quantum market. Since each event projector $\Pi^A_{a|x}\otimes\Pi^B_{b|y}$ is rank one, these assets are quantum Arrow securities under~\cref{def:quantum-arrow-security}. For each fixed input pair $(x,y)$, the four event projectors form a complete family of securities corresponding to the outcomes of a common projective measurement.

\subsection{Exclusive assets}\label{subsec:exclusive-assets}
We generalize our asset models, both classical and quantum, to what we call an \emph{exclusive asset model}. Such a model naturally lends itself to being described by tools from quantum contextuality; see~\cref{subsec:exclusivity-graph-chsh}. Nevertheless, such a model is also classically well-motivated, such as in the bond example we mentioned earlier, where an investor would not bet on the price of a bond going both up and down. Consider the formal definition of exclusivity for random variables and, by extension, for classical assets in~\cref{def:excl-assets}. We extend this notion to quantum assets to define \emph{exclusive quantum assets}:

\begin{definition}[Exclusive quantum assets]\label{def:excl-assets-quantum}
    Two quantum assets $A$ and $B$ are said to be \emph{exclusive} if $AB = 0$.
\end{definition}

That is, two quantum assets are exclusive when their nonzero payoff subspaces are orthogonal. We can succinctly describe the exclusive assets model via an \emph{asset exclusivity graph} $G(\vb{A}, E)$:

\begin{definition}[Asset exclusivity graph]\label{def:asset-excl-graph}
    Given a market populated with assets, denote a basket of assets by $\vb{A} = \parens*{A_1, A_2, \ldots, A_{d}}$, where $d \in \NN$. Then, an \emph{asset exclusivity graph} $G(\vb{A}, E)$ is a graph with vertex set $\vb{A}$ and edge set $E$, where two assets $A_i, A_j \in \vb{A}$ share an edge in $E$ if and only if $A_i$ and $A_j$ are exclusive assets.
\end{definition}

By considering the exclusivity graph, we can use some tools and results from \emph{quantum contextuality}~\cite{pavicic2023quantum,budroni2022kochen-specker}. For example, Cabello \etal~\cite{cabello2013basic} showed that quantum violations of non-contextual inequalities require exclusivity graphs that contain induced odd cycles with 5 or more vertices, and/or their complements. This identifies the graph structures that can support a quantum-classical separation, which we quantify in~\cref{subsec:exclusive-asset-portfolios} in the setting of portfolios using the independence number and the \lovasz theta number. The two theorems that we present in the next section are generalizations of Ref.~\cite[Result 1]{cabello2014graph}.

\subsection{Exclusive asset portfolios}\label{subsec:exclusive-asset-portfolios}
We now combine the definitions and basic concepts that we introduced in~\cref{subsec:classical-assets} with the notion of asset exclusivity from~\cref{subsec:exclusive-assets} to determine the maximum value of portfolios comprised of classical and quantum assets. The notion of exclusivity is well-motivated not just in classical finance as mentioned previously, but also in quantum portfolio optimization, where there has been recent work~\cite{yalovetzky2026quantum-informed} using quantum algorithms to solve the portfolio selection problem in graph-theoretic language.

We start with our first result connecting the maximum value of a portfolio comprised of bounded random variables to the weighted independence number of the asset exclusivity graph $G(\vb{A}, E)$. Recall that an independent set of the asset exclusivity graph represents a set of nonexclusive assets. The result is that the independence number is the upper bound for the expected portfolio value when maximizing over probability distributions, and is in fact tight given specific sets of assets.

\begin{thm}[Maximum classical portfolio value]\label{thm:max-classical-portfolio-value}
    Let $\vb{A} = (A_1, A_2, \ldots, A_d)$ be a collection of random variables describing financial assets. Each random variable is bounded as {$A_i \leq 1$}. Let $G(\vb{A}, E)$ be the  exclusivity graph of these assets and let $\vb{w} \in \RR_+^d$ be a portfolio. Then, the expected portfolio value maximized over probability distributions is upper bounded by the weighted independence number $\alpha(G,\vb{w})$. For a given exclusivity graph, there exist assets $A_1, A_2, \ldots, A_d$, each satisfying $A_i \leq 1$, such that the bound is saturated. 
\end{thm}

\begin{proof}
    We begin by proving the upper bound, then we prove saturation. Since
        \begin{align}
        \mathbb{E}_p[\mathbf{w} \cdot \mathbf{A}] = \sum_{\omega \in \Omega}p(\omega)\sum_{i=1}^d w_iA_i(\omega) ,
    \end{align}
    it suffices to prove that
    \begin{equation}
        \defsum{i}{}{w_iA_i(\omega)\leq \alpha(G, \vb{w})}
    \end{equation}
    for every $\omega \in \Omega$. Define auxiliary assets $\curly*{\Tilde{A}_i}$ as
    \begin{align}
        \Tilde{A}_i(\omega) =
        \begin{cases}
            0 & \text{if } A_i(\omega) = 0,\\
            1 & \text{otherwise.}
        \end{cases}
    \end{align}
    Then, the exclusivity graph of $\curly*{\Tilde{A}_i}$ is the same as that of $\curly*{A_i}$. Moreover, since $A_i \leq \Tilde{A}_i$, we have that
    \begin{equation}
        \defsum{i}{}{w_iA_i(\omega)} \leq \defsum{i}{}{w_i\Tilde{A}_i(\omega)}.
    \end{equation}
    For binary-valued assets, it holds that
    \begin{equation}
        \defsum{i}{}{w_i\Tilde{A}_i(\omega)} = \defsum{i : \Tilde{A}_i(\omega) = 1}{}{w_i}.
    \end{equation}
    Let $S(\omega)\coloneqq \curly*{i : \Tilde{A}_i(\omega)= 1}$. For $i,j \in S(\omega)$, $\Tilde{A}_i(\omega)\Tilde{A}_j(\omega) = 1 \neq 0$ so $(i,j) \notin E$. Hence, $S(\omega)$ is an independent set. By the definition of the weighted independence number,
    \begin{equation}
        \defsum{i\in S(\omega)}{}{w_i\leq \alpha(G, \vb{w})},
    \end{equation}
    which concludes the upper bound.

    To prove saturation, let $G = (V,E)$ be a given graph and choose the sample space $\Omega_G$ to be the set of all independent sets of $G$. For each vertex $i \in V$, we define the binary asset
    \begin{equation}
        A_i(S) \coloneqq
        \begin{cases}
            1, & i \in S,\\
            0, & i \notin S,
        \end{cases}
        \qquad S \in \Omega_G.
    \end{equation}
    These assets have exclusivity graph exactly $G$. If $(i,j) \in E$, no independent set contains both vertices, so $A_i(S)A_j(S) = 0$ for every $S \in \Omega_G$. On the other hand, if distinct vertices $i,j$ are nonadjacent, then $\{i,j\} \in \Omega_G$, and both assets equal $1$ on this outcome. Now, we choose an independent set $S^\star$ achieving the weighted independence number:
    \begin{equation}
        \defsum{i \in S^\star}{}{w_i} = \alpha(G, \vb{w})
    \end{equation}
    We denote by $p^\star$ the probability distribution concentrated on $S^\star$. Then,
    \begin{equation}
        \EE_{p^\star}[\vb{w} \cdot \vb{A}] = \defsum{i=1}{d}{w_i A_i(S^\star)} = \defsum{i \in S^\star}{}{w_i} = \alpha(G, \vb{w}),
    \end{equation}
    which proves saturation and thus completes the proof.
\end{proof}

One may associate the probability distribution of the previous theorem with a diagonal density matrix. We can  denote this density matrix by $\rho_\mathrm{cl}$ and call it a \emph{classical market state}; see also~\cref{def:market-density-operator}. The next theorem presents the result for quantum portfolios, that is, portfolios defined by quantum assets/operators, and maximizing over general density operators; see~\cref{def:market-density-operator}. Compared to earlier work~\cite{cabello2014graph}, the following result is extended beyond rank-1 projectors. 

\begin{thm}[Maximum quantum portfolio value]\label{thm:max-quantum-portfolio-value}
    Given $d \in \NN$ quantum {assets} $\vb{A} = \lp A_1, A_2, \ldots, A_{d} \rp$ bounded between $0$ and $I$, asset exclusivity graph $G(\vb{A}, E)$, and portfolio $\vb{w} \in \RR_+^d$, the expected portfolio value maximized over all mixed quantum states is upper bounded by the weighted \lovasz theta number $\vartheta(G, \vb{w})$. For a given exclusivity graph, there exist assets $A_1, A_2, \ldots, A_d$, each satisfying $0 \preceq A_i \preceq I$, such that the bound is saturated.
\end{thm}

\begin{proof}
    We notice that, for a given set of quantum assets, the expected portfolio value is maximized when the quantum state is a projector into the subspace corresponding to the largest eigenvalue of the weighted sum of the assets:
    \begin{align}
        \EE_\rho[\vb{w}\cdot \vb{A}] = \Tr[\rho\lp \vb{w} \cdot \vb{A} \rp] \leq \lambda_{\max}\lp\vb{w} \cdot \vb{A}\rp.
    \end{align}
    Therefore, we need to show that this eigenvalue is upper bounded by the weighted \lovasz theta number. To this end, we use the following formulation of the weighted \lovasz theta number~\cite{galli2017lovasz}. Let $B$ range over all $d\times d$ Hermitian positive semidefinite matrices such that $B_{ij}=0$ for $(i,j)\in E$ and such that $\Tr[B]=1$. Then,
    \begin{align}
        \vartheta(G, \vb{w})= \max_B \defsum{i,j=1}{d}{\sqrt{w_i w_j} B_{ij}}.
    \end{align}
    We will construct a matrix $B$ that satisfies the above conditions and for which the right-hand side is lower bounded by $\lambda_{\max}\parens*{\vb{w} \cdot \vb{A}}$. Let $\ket{\psi} \equiv \ket{\psi(\lambda_{\max})}$ be the eigenvector of $\vb{w} \cdot \vb{A}$ with this eigenvalue and denote $\ket{\psi_i} = \sqrt{w_i} A_i\ket{\psi}$. Define a $d \times d$ matrix $X$ with elements $X_{ij} = \braket{\psi_i}{\psi_j}$. By construction, $X$ is a Hermitian, positive semidefinite matrix and $X_{ij} = 0$ for $(i,j) \in E$. Using the fact that each operator $A_i$ is bounded between $0$ and $I$, we have
    \begin{align}
        \Tr[X] = \defsum{i=1}{d}{\braket{\psi_i}} = \mel{\psi}{\defsum{i=1}{d}{w_iA_i^2}}{\psi} \leq \mel{\psi}{\vb{w} \cdot \vb{A}}{\psi} = \lambda_{\max}(\vb{w} \cdot \vb{A}).
    \end{align}
    Moreover,
    \begin{align}
        \defsum{i,j=1}{d}{\sqrt{w_iw_j}X_{ij}} = \defsum{i,j=1}{d}{\sqrt{w_iw_j}\braket{\psi_i}{\psi_j}} = \norm{\defsum{i=1}{d}{\sqrt{w_i}\ket{\psi_i}}}_2^2 = \lambda_{\max}^2(\vb{w} \cdot \vb{A}).
    \end{align}
    If $\lambda_{\max}(\vb{w} \cdot \vb{A}) = 0$, the upper bound is immediate; otherwise, $\Tr[X] > 0$. To obtain the matrix $B$, we normalize $X$ by its trace: $B = X/\Tr[X]$. Then,
    \begin{equation}
        \defsum{i,j=1}{d}{\sqrt{w_iw_j}B_{ij}} = \frac{\lambda_{\max}^2(\vb{w} \cdot \vb{A})}{\Tr[X]} \geq \lambda_{\max}(\vb{w} \cdot \vb{A}) \implies \lambda_{\max}(\vb{w} \cdot \vb{A}) \leq \vartheta(G,\vb{w}), 
    \end{equation}
    which proves the upper bound.

    To prove saturation, choose an optimal matrix $B_*$ in the above maximization, which exists because the feasible set is nonempty and compact. Take a Gram decomposition $(B_*)_{ij} = \braket{x_i}{x_j}$. Then, $\Tr[B_*] = \defsum{i}{}{\braket{x_i}} = 1$ and $\ket{\phi} = \parens*{\defsum{i}{}{\sqrt{w_i}\ket{x_i}}} / \sqrt{\vartheta(G,\vb{w})}$ is a unit vector. Let $A_i = \ketbra{x_i}/\braket{x_i}$ when $\ket{x_i} \neq 0$, and $A_i = 0$ otherwise. These projectors satisfy $A_i A_j = 0$ for $(i,j) \in E$.

    Set $r_i = \sqrt{\braket{x_i}}$ and $c_i = \sqrt{w_i}\braket{\phi}{x_i}/r_i$ for $r_i > 0$, with $c_i = 0$ when $r_i = 0$. Then $\defsum{i}{}{r_i^2} = 1$ and $\abs{c_i}^2 = w_i\mel{\phi}{A_i}{\phi}$. By the definition of $\ket{\phi}$ and Cauchy--Schwarz applied to $(r_i)_i$ and $(c_i)_i$,
    \begin{align}
        \vartheta(G,\vb{w})
        &= \abs{\defsum{i}{}{\sqrt{w_i}\braket{\phi}{x_i}}}^2
         = \abs{\defsum{i}{}{r_i c_i}}^2\\
        &\leq \parens*{\sum_i r_i^2}\parens*{\sum_i\abs{c_i}^2}
         = \defsum{i}{}{w_i\mel{\phi}{A_i}{\phi}}.
    \end{align}
    Taking $\rho = \ketbra{\phi}$, the upper bound gives $\Tr[\rho(\vb{w}\cdot\vb{A})]=\vartheta(G,\vb{w})$.
\end{proof}

We note that the pure market state $\rho_{\lambda_{\max}} \coloneqq \ketbra{\psi(\lambda_{\max})}{\psi(\lambda_{\max})}$ is constructed from the eigenvector of the portfolio $\vb{w} \cdot \vb{A}$ and achieves the largest expected payoff for the portfolio from the variational principle. Therefore, we call $\rho_{\lambda_{\max}}$ the \emph{maximum eigenmarket} of the portfolio. To relate~\cref{thm:max-classical-portfolio-value,thm:max-quantum-portfolio-value}, we highlight a well-known result from \lovasz~\cite{lovasz1979shannon}:

\begin{fact}\label{fact:ind-lovasz-numbers}
    Let $G$ be a graph. Then, $\alpha(G) \leq \vartheta(G)$.
\end{fact}

An immediate consequence of~\cref{thm:max-classical-portfolio-value,thm:max-quantum-portfolio-value,fact:ind-lovasz-numbers} is a simple relationship between the maximum portfolio value in a classical market and a quantum market. In particular, given an asset exclusivity graph $G$, we have $\alpha(G) \leq \vartheta(G)$. In many scenarios, this is a strict inequality: $\alpha(G) < \vartheta(G)$. Thus, the expected value of a portfolio comprised of quantum assets can attain an expected value greater than that of a portfolio of classical assets. In~\cref{sec:chsh-market-game}, we will see that if traders in a quantum market have a high belief in the ability of an exchange to produce entangled quantum states, then they will estimate a higher expected return and will thus pay more to participate in the market. Before developing the full market, we illustrate these notions with a simple example:

\begin{example}[Simplest asset exclusivity graph with quantum advantage]\label{ex:simplest-asset-excl-advantage}
    Consider a scenario with five assets, $\vb{A} = \{A_1, A_2, A_3, A_4, A_5\}$. If the asset exclusivity graph for these assets is the cycle graph $C_5$ as in~\cref{fig:five-asset-excl-graph}, then the maximum value of a portfolio of classical assets is $\alpha(C_5) = 2$ while the maximum value of a portfolio of quantum assets is $\vartheta(C_5) = \sqrt{5} \approx 2.24$ (see~\cref{subsec:graph-theory} for analytic expressions for the independence and \lovasz theta numbers of cycle graphs). Thus, we see that $\alpha(C_5) < \vartheta(C_5)$. This is the simplest asset exclusivity graph that yields a quantum-classical separation.
\end{example}

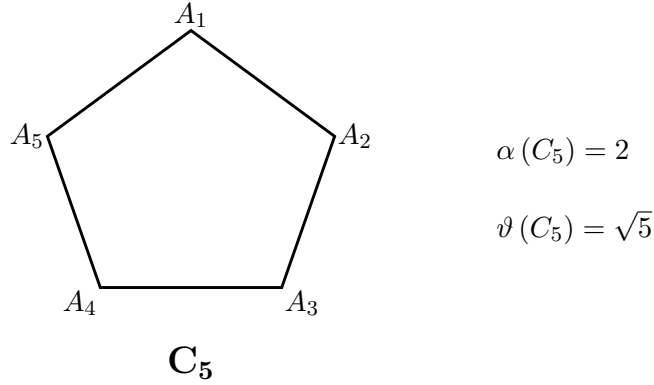
\begin{figure}[tb]
    \centering
    \begin{tikzpicture}[every node/.style={inner sep=1pt}]

    \coordinate (A1) at (0,2.2);
    \coordinate (A2) at (1.9,0.8);
    \coordinate (A3) at (1.2,-1.2);
    \coordinate (A4) at (-1.2,-1.2);
    \coordinate (A5) at (-1.9,0.8);

    \draw[line width=1.1pt] (A1) -- (A2) -- (A3) -- (A4) -- (A5) -- cycle;

    \node[above]      at (A1) {$A_1$};
    \node[right]      at (A2) {$A_2$};
    \node[below right] at (A3) {$A_3$};
    \node[below left]  at (A4) {$A_4$};
    \node[left]       at (A5) {$A_5$};

    \node at (0,-2.2) {\Large $\mathbf{C_5}$};

    \node[anchor=west] at (4.0,0.6) {$\alpha\left(C_5\right)=2$};
    \node[anchor=west] at (4.0,-0.4) {$\vartheta\left(C_5\right)=\sqrt{5}$};

    \end{tikzpicture}
    \caption{The simplest asset exclusivity graph with quantum advantage, a cycle graph, $C_5$ with independence number $\alpha(C_5) = 2$ and \lovasz theta number $\vartheta(C_5) = \sqrt{5} \approx 2.24$.}
    \label{fig:five-asset-excl-graph}
\end{figure}

We conclude this section with a few remarks. First, we emphasize that only those markets populated by a set of quantum assets whose exclusivity graph contains an induced odd cycle of length at least five, or its complement, as an induced subgraph will exhibit a quantum-classical separation; \cref{ex:simplest-asset-excl-advantage} is thus a minimal example for a separation. Second, we highlight a result about random graphs, namely that the ratio of the \lovasz theta number to the independence number, $\vartheta(G)/\alpha(G)$, is an unbounded function of the number of vertices in the graph $G$. In particular, for Erd\H{o}s-R\'{e}nyi random graphs\footnote{An Erd\H{o}s-R\'{e}nyi graph $G(n,p)$ is a random graph with $n$ vertices where each of the $\frac{n(n-1)}{2}$ possible edges in the graph exists with probability $p$.}, several early works~\cite{matula1970complete,grimmett1975colouring,bollobas1976cliques,frieze1990independence} showed that the independence number scales as $\bigTheta{\log{n}}$ and Juh\'{a}sz~\cite{juhasz1982asymptotic} showed that the \lovasz theta number scales as $\bigTheta{\sqrt{n}}$. Thus, the ratio of the \lovasz theta number to the independence number, $\vartheta(G)/\alpha(G)$, of an Erd\H{o}s-R\'{e}nyi random graph $G$ drawn from $G(n,p)$ with fixed $p \in (0,1)$ scales as $\bigTheta{\sqrt{n}/\log(n)}$ with high probability as $n \to \infty$. Thus, if we define a universe of quantum assets and model their exclusivity graph as an Erd\H{o}s-R\'{e}nyi graph, the separation between the classical and quantum portfolio values grows as the number of assets that we choose grows. These results provide a route to constructing large asset collections with growing quantum–classical separation. Extending the market structure and pricing mechanism developed below in~\cref{sec:chsh-market-game} to such collections is an interesting direction for future work.

\section{CHSH market game}\label{sec:chsh-market-game}
We now put everything together into a fully functioning market and see how prices, trades, and other financial features emerge. We instantiate our market using a \emph{nonlocal quantum game}, specifically the CHSH game (see~\cref{subsec:chsh-game}). While we develop our quantum market using the CHSH game, the principles behind the market extend to other nonlocal games~\cite{brunner2014bell}, such as other \textsc{xor} games~\cite{cirelson1980quantum,cleve2004consequences,briet2009generalized}, or games based on the $I_{3322}$~\cite{collins2004relevant} or Svetlichny~\cite{svetlichny1987distinguishing} inequalities.

In this section, we systematically build up the details of the game, starting with the specification of the financial contract binding all parties to the rules in~\cref{subsec:contract-specification}, where we also describe the structure of the market. Then, in~\cref{subsec:market-state-supply}, we describe how the market state is produced and how each party characterizes the state, which we call suggestively the "market research phase". We put everything together in~\cref{subsec:strategy-choice-pricing}, presenting the main results on the strategy selection and the pricing mechanism that each party follows. Finally, in order to build intuition for how the market works, we provide a numerical example in~\cref{subsec:numerical-examples}.

\subsection{Financial contract}\label{subsec:contract-specification}
We begin by introducing the \emph{financial contract} for the CHSH market game. The contract describes everything about the market game: the structure of the market, the rules of the game and the winning condition, and the tradable assets. We go through each of these in this section; the financial contract can be thought of as a summary of what we discuss here, which all parties involved sign and bind themselves to. We start by discussing \emph{who} these parties are, that is, defining the \emph{market participants}.

\subsubsection{Market participants}\label{subsubsec:market-participants}
\begin{figure}[tb]
    \centering
    \includegraphics[width=\linewidth]{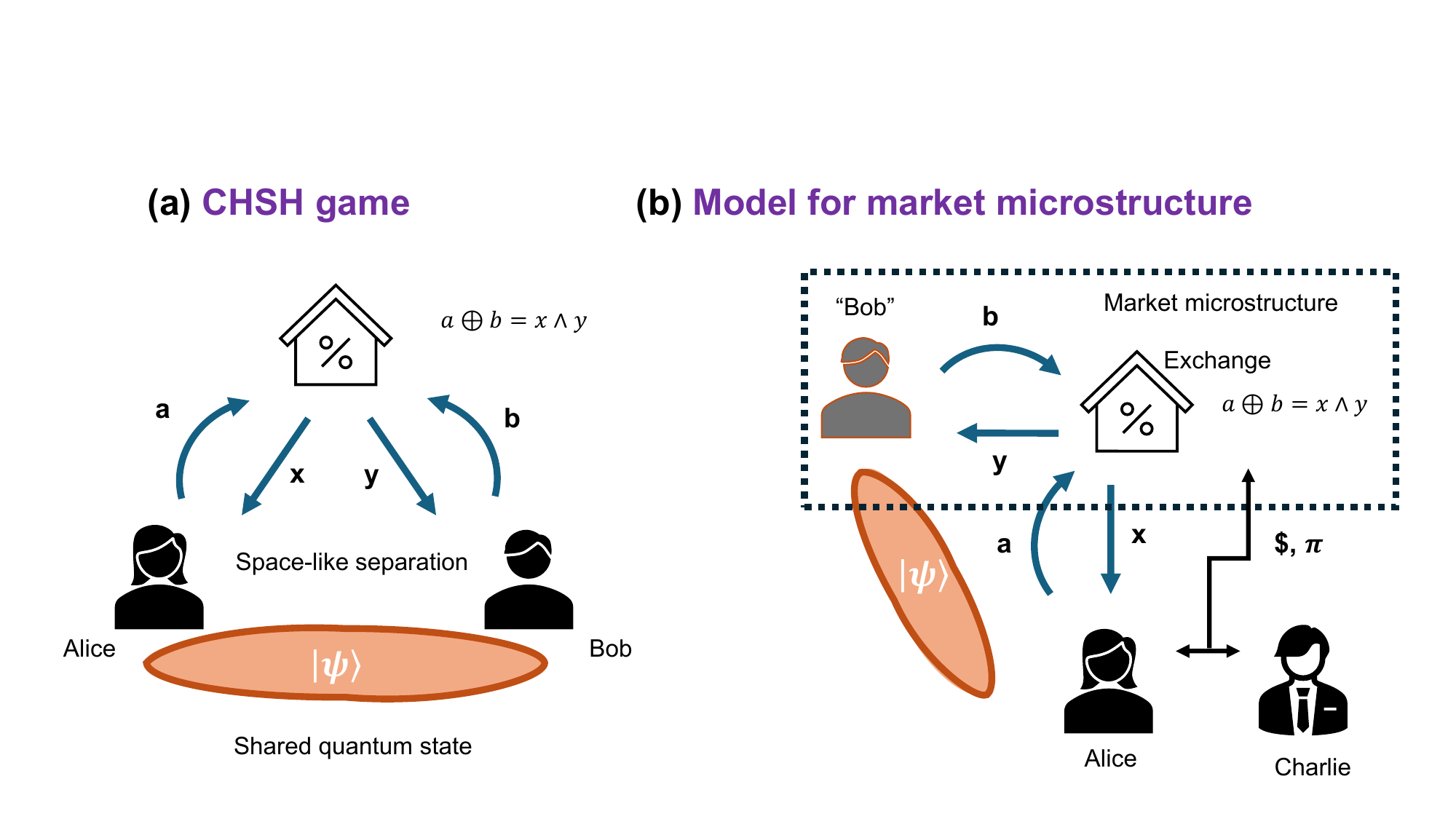}
    \caption{(a) Sketch of the standard CHSH game. Two parties, Alice and Bob, share a quantum state. They receive classical bits $x$ and $y$ and are tasked with returning the classical bits $a$ and $b$ such that a predicate $a \oplus b = x \wedge y$ is satisfied. (b) Sketch of the interpretation of the CHSH game as a market microstructure. The trader Alice shares an entangled state with the exchange, which has subsumed the party we normally call Bob. The actions of Bob are therefore part of the market microstructure and are taken as those in the standard CHSH game. Also depicted is the exchange-mediated trade of outcome-contingent claims between Alice and Charlie, together with the associated cash transfers.}
    \label{fig:chsh-game-setup}
\end{figure}

There are three parties involved in the market: Alice, Charlie, and a neutral third party which we call the \emph{exchange}, which we assume subsumes who we normally call Bob in the CHSH context. For a summary, we refer the reader to~\cref{fig:chsh-game-setup}, which shows the market structure consisting of Alice and the exchange (plus Bob), while omitting the standard, classical interactions with Charlie. Alice and Charlie are the economic counterparties of the trade, and the exchange serves as a trusted referee and settlement authority with various responsibilities that we detail below. We call Alice and Charlie \emph{traders} throughout the rest of the paper, and consider Alice to be on the long/buy side and Charlie on the short/sell side of the trades in the market. Alice wants to maximize her wealth, while Charlie wants to maximize his wealth. Alice takes the long/buy side of the trade, so she pays to play the game and purchase quantum assets to help her win the game. On the other side of the trade is Charlie, who takes the short/sell side. Charlie initially receives some premium from Alice and either keeps the premium or must pay up, depending on the result of the game, that is, whether or not the winning condition (i.e., the predicate) is satisfied.

The exchange has no financial stake in the contract and exists only to facilitate the playing of the game. In addition, the exchange shares the necessary quantum resources with Alice, performs the necessary local measurements on the exchange's part of the quantum resources to play the role of Bob in the usual CHSH game, and enforces any no-communication constraints/spacelike separation required from the quantum game. The exchange therefore has the following responsibilities:
\begin{enumerate}
    \item Acting as an intermediary for communication between Alice and Charlie to agree on a trade price of the game. Alice submits her bid price and Charlie submits his ask price, and these prices are adjusted until agreement is reached, after which the exchange commits both sides to the trade price. If no trade price can be agreed upon, then the exchange aborts the negotiation between Alice and Charlie.
    \item Producing the market state $\rho$. The market state is fixed by the exchange and the traders have complete or partial information about this state (see~\cref{subsec:market-state-supply}). In contrast, in the CHSH game, Alice and Bob are allowed to share a known quantum state as part of their strategy. 
    \item Sending Alice her respective qubit, while retaining Bob's qubit. This simulates the spacelike separation because the exchange, being neutral, abides by the no-communication rule in the usual CHSH game, having no incentive to cheat.
    \item Sampling the input bits $x$ and $y$ according to some distribution and sending $x$ to Alice and retaining $y$ for Bob.
        \item Receiving the output bit $a$ from Alice and producing the output $b$ for Bob using Bob's part of the shared quantum state. Bob's measurement is fixed in advance by the agreed-upon strategy and is carried out upon receipt of the input bit $y$, independently of Alice's answer $a$.
        \item Checking whether the predicate $a \oplus b \overset{?}{=} x \wedge y$ is satisfied and, if so, declaring the winner.
        \item Clearing the settlement: exchanging the appropriate amount of capital between Charlie and Alice if Alice wins and allowing Charlie to keep the premium if Alice loses.
\end{enumerate}

We will discuss these points in more detail throughout this section as we develop the market. As mentioned above, we assume that the exchange is neutral and honest; it does not cheat or deviate from its responsibilities. This allows us to set aside adversarial strategies that either side may take if, for example, the CHSH market game were administered directly between the two sides of the trade. For example, if Charlie were the one producing the market state, sending and receiving the bits, checking the predicate, etc., it would be very easy for him to lie about the outcome of the game and unfairly keep the premium that Alice paid to play the game, even if she ended up winning the game. While there may be interesting cryptographic (classical or quantum) mechanisms that one could devise to enhance the security of markets based on nonlocal games or to implement such a market without a trusted exchange, that is not our focus in this paper.

We also stress that the neutrality of the exchange precludes it from cheating and helping Alice win. The simulated Bob that exists inside of the exchange to facilitate the playing of the CHSH game with Alice operates according to the rules specified by the contract. That is to say, in performing the actions required by the CHSH game part of the market, where Bob receives his input bit and measures his part of the shared quantum state in a pre-specified basis according to the agreed-upon strategy, Bob/the exchange act as if they are spacelike-separated from Alice. This prevents any communication between Alice and Bob after the game has started (i.e., after the input bits have been received), so coordinating the output is precluded. 

\subsubsection{Multi-round CHSH market game and settlement}\label{subsubsec:chsh-market-game-rules}
In this market game, as mentioned above, Alice is playing the CHSH game, facilitated by the exchange, who performs all of the usual actions that Bob and the referee would in the CHSH game where Alice and Bob are space-like separated. We make some modifications and add some financial machinery to make a market out of the CHSH game. First, as stated before, it is only Alice who is playing the game, with the honest exchange simulating a "Bob" doing all of the usual things Bob does in the CHSH game. Second, we allow for more general distributions from which the input $(x,y)$ is drawn and denote the probability distribution by $\pi(x,y)$. This distribution is either known in advance or can be estimated from prior market data or a pre-game calibration phase; the distribution affects the portfolio of quantum assets that Alice constructs before playing. In this work, we assume that the input distribution is known, but we mention in~\cref{sec:discussion} the potential for future work in markets where the distribution is unknown and must be learned.

Furthermore, while in~\cref{subsubsec:market-participants} we described one round of the CHSH game, we assume that the market can be instantiated via the playing of $T$ rounds of the game, where for each round $t \in [T]$, a fresh, independently-prepared copy of the market state and a new, independently-sampled input $(x_t, y_t)$ are provided by the exchange. Then, the per-round payout function is defined as
\begin{equation}
    S_t \coloneqq
    \begin{cases}
        1 & \text{if $a_t \oplus b_t = x_t \wedge y_t$}\\
        0 & \text{otherwise.}
    \end{cases}
    \label{eqn:per-round-payout}
\end{equation}
After playing $T$ rounds, Charlie pays Alice
\begin{equation}
    S = \defsum{t=1}{T}{S_t}.
\end{equation}
Thus, Charlie's maturity settlement cashflow is $-S$, and the exchange acts as the settlement authority to ensure that the payout is settled correctly.

\subsubsection{Listed financial instruments and portfolio mechanism}\label{subsubsec:listed-instruments}
We have now established \emph{who} populates the quantum market game and \emph{what} the game is, so now we specify what the tradable assets are. That is, we define the financial instruments involved in the trade that occurs between Alice and Charlie. Here, we connect the market game with the quantum assets that we defined in~\cref{subsec:quantum-assets-chsh}. We take the winning projectors from~\cref{fig:chsh-excl-graph},
\begin{equation}
    \Pi_k \coloneqq \Pi_{a|x}^A \otimes \Pi_{b|y}^B,\label{eq:def-quantum-assets}
\end{equation}
to be the assets that Alice can buy, where $k \in \mathcal{K}$ and where
\begin{equation}
    \mathcal{K} \coloneqq \{(a,b \mid x,y) \in \{0,1\}^4 : a \oplus b = x \wedge y\}
\end{equation}
denotes the set of input-output bits for the winning combinations. Thus, there are $\abs{\mathcal{K}} = 8$ individual assets. In this market, the rational thing for Alice to do is to purchase and use the two winning-event assets associated with each input pair together, which we call a \emph{bundle}. A bundle for input pair $(x,y)$ thus contains one unit of each asset indexed by $\mathcal{K}_{(x,y)}$, where
\begin{equation}
    \mathcal{K}_{(x,y)} \coloneqq \{(a,b \mid x,y) \in \mathcal{K}\},
\end{equation}
and where $\abs{\mathcal{K}_{(x,y)}} = 2$. For example, a bundle for $(x,y) = (0,0)$ contains one unit of $\Pi_{(0,0|0,0)}$ and one unit of $\Pi_{(1,1|0,0)}$. Before the game begins, Alice purchases $N$ projectors from Charlie (or $N/2$ bundles in total; we assume Alice is rational), subject to her budget $\xi$, where the portfolio is $\vb{w} = (w_k)_{k \in \mathcal{K}}$, where $w_k \in \ZZ_{\geq 0}$ counts units of the individual assets $\Pi_k$, and
\begin{equation}
    \defsum{k \in \mathcal{K}}{}{w_k} = N.
\end{equation}
In a given round, the exchange samples the input $(x,y)$, sends $x$ to Alice, and retains $y$ for Bob in the exchange. The exchange checks whether a matching bundle remains in Alice's inventory and, if so, it authorizes a quantum round, in which Alice and Bob perform the complete local projective measurements
\begin{equation}
    \{\Pi_{0|x}^A, \Pi_{1|x}^A\},\qquad \{\Pi_{0|y}^B, \Pi_{1|y}^B\},
\end{equation}
respectively, on a fresh copy of the market state. Each party's measurement depends only on its own input and the joint outcomes comprise all four possible $(a,b) \in \{0,1\}^2$. If the output satisfies $a \oplus b = x \wedge y$, then exactly one claim in the bundle pays one unit of the numeraire; otherwise, neither claim pays. Thus, the bundle contributes exactly the per-round payout $S_t$ to the aggregate settlement paid at the end of the game, as defined in~\cref{eqn:per-round-payout}. Both asset units are removed from the inventory after the round, regardless of its outcome. Each quantum round therefore uses one bundle and one shared copy of the market state, and if no matching bundle remains, the round is played using the prespecified classical strategy. Any unused units expire without payment after the final round. We also note that the exchange enforces these inventory and settlement rules, as we assume that the exchange is trusted to implement the prescribed local procedures and prevent unauthorized use of the quantum strategy. Developing cryptographic mechanisms to enforce these rules without a trusted exchange is an interesting direction for future work.

\subsection{Market state supply and information access}\label{subsec:market-state-supply}
Essential to our market is the market state, which we defined above in~\cref{def:market-density-operator}. This state is prepared by the exchange, and one qubit is distributed to Alice and the other qubit is kept by the exchange to simulate the "Bob" side of the CHSH game. However, exact state preparation is not a realistic assumption in either current or future hardware proposals. Hence, we assume that the state preparation is imperfect and the exchange prepares the relevant entangled state (e.g., the Bell state in the CHSH game) with probability $1-\eta$ and the maximally-mixed state with probability $\eta$. As we are basing our market on the CHSH game, we take the market state to be
\begin{equation}
    \rho(\eta) = (1-\eta)\ket{\Phi^+}\bra{\Phi^+} + \frac{\eta}{4}I_4,\label{eq:noisy-market-state}
\end{equation}
where $\eta \in [0,1]$, $\ket{\Phi^+} = \frac{1}{\sqrt{2}}(\ket{00} + \ket{11})$ is a Bell state, and $I_4$ is the identity matrix on two qubits. Such a state is called a \emph{Werner state} and is of independent interest, especially in the contextuality literature~\cite{werner1989quantum,horodecki1995violating,barrett2002nonsequential,acin2006grothendiecks}. We choose the Werner state as it is a simple state with which we can demonstrate the dynamics of the market we develop; there is nothing fundamental about this choice, and interesting future work could consider other classes of states. The value $\eta$ is related to, in contextuality terms, the \emph{local fraction} and the \emph{contextual fraction}~\cite{elitzur1992quantum,abramsky2017contextual}, which tells us how much our model is weighted towards a local/non-contextual model versus a nonlocal/contextual model, but we do not use this terminology here. Instead, we refer to $\eta$ simply as a noise parameter. Thus, with probability $1-\eta$, Alice has access to an entangled quantum state and with probability $\eta$, the state only contributes white noise. Thus, a rational trader must further discount their investment by an amount that depends on $\eta$. We consider three scenarios of knowledge of $\eta$:
\begin{enumerate}
    \item Both Alice and Charlie know $\eta$ exactly;
    \item Both Alice and Charlie agree on an incorrect value for $\eta$;
    \item Alice and Charlie disagree on the value for $\eta$.
\end{enumerate}
In order to get access to this information, we allow for a "market research" phase before the game begins. That is, each side of the trade is allowed to purchase copies of the quantum market state $\rho(\eta)$, with which they are allowed to characterize the state using any protocol of their choosing, such as tomography, parameter estimation, or self-testing~\cite{supic2020self-testing}. To be in scenario 1 above, either the exchange simply broadcasts the true value of $\eta$ or Alice and Charlie each use the same characterization protocol and/or device and, in the limit of an infinite number of copies of the state, arrive at the same (correct) value for $\eta$. To be in scenario 2, Alice and Charlie use the same flawed or noisy device that, in the asymptotic limit of the number of states, provides an incorrect estimate for $\eta$. Finally, scenario 3 can arise in several regimes, such as when Alice and Charlie have access to a different number of copies of the state, use characterization devices/protocols of differing quality, or both, and thus arrive at different estimates for $\eta$. In fact, even with the same (finite) number of copies and the same characterization device/protocol, they can still arrive at different values for $\eta$ simply from the random classical and quantum fluctuations in the Werner state, finite-sample error, etc. Thus, scenario 3 is arguably the most plausible one to consider.

To realize this information access in a financial way, the exchange \emph{charges} the traders some amount of capital per copy of the state. Such a model is realistic, as it could cost some amount of money for the exchange to produce quantum states, especially entangled quantum states. Then, there arises a tradeoff between the quality of information each side has on the true state of the market and the amount each side can invest to actually play the game. For example, Alice may spend a lot of her capital in buying copies of the market state to get access to better information, but at the cost of being able to purchase fewer quantum assets to play the game. This idea is well-motivated financially as well, as even in classical markets, market participants must pay for high-quality data that can better inform their investment decisions, but at the cost of subsequently having less capital with which to invest.

Finally, we highlight that the market state is assumed not to change during the market research phase or the playing of the game. That is, the research copies and the game-round copies are drawn from the same stationary preparation procedure. We leave to future work how a changing market state affects the pricing dynamics and strategy choice.

\subsection{Multi-round strategy choice and pricing analysis}\label{subsec:strategy-choice-pricing}
We now analyze strategy selection and valuation of the quantum assets. As we will see, the structure of the game determines the price of the quantum portfolio from first principles. Hence, we provide a novel perspective on quantum games in terms of quantum assets, and, in particular, we connect this with the exclusive quantum assets that we introduced in~\cref{sec:assets}.

\subsubsection{Fixed quantities and portfolio cost}\label{subsubsec:strategy-space}
We summarize what quantities and objects are fixed by the contract and/or the exchange. First, the exchange fixes the input distribution $\pi(x,y)$, which we assume is known by everyone. Next, the predicate that must be satisfied in order for Alice to win the game is stated and agreed upon in the contract. Next, all parties agree to the total number of rounds $T$ of the CHSH game that will be played and how the final settlement will be reached (i.e., the payout per round, the aggregate payout after $T$ rounds, and the cashflow before and after all rounds are played). All parties also agree to the form of the market state $\rho(\eta)$, even if the true value of $\eta$ is not known by one or both sides of the trade. In the case of the CHSH market we are developing, both sides agree that $\rho(\eta)$ is a Werner state of the form given in~\cref{eq:noisy-market-state}.

Given all of the information that each side has access to, each side must assess the value of participating in the market with a specified portfolio $\vb{w} = (w_k)_{k \in \mathcal{K}}$. These assessments determine the prices Alice and Charlie are willing to pay or accept. For an agreed trade, let $c_k$ denote the upfront charge per unit of the individual asset $\Pi_k$. The total amount Alice pays Charlie, through the exchange, is then
\begin{equation}
    C(\portvec) = C_0 + \underbrace{\defsum{k \in \mathcal{K}}{}{c_k \port_k}}_{C_Q(\vb{w})},\label{eq:portfolio_cost}
\end{equation}
where $C_0$ is an entry premium and $C_Q(\vb{w})$ is the additional charge for the quantum-asset portfolio. Recall that Alice purchases the quantum assets in bundles, where each bundle contains one unit of each of the two assets indexed by $\mathcal{K}_{(x,y)}$. Thus, we can consider the price of a bundle, which is $\defsum{k \in \mathcal{K}_{(x,y)}}{}{c_k}$.

Note that we include the entry premium $C_0$ because if Alice decided to not buy any quantum assets and just played the game classically, she would have a free positive expected payoff. Note, however, that $C_0$ is still charged even if Alice buys enough projectors such that she is guaranteed to not have to use the classical strategy. If Alice has a budget $\xi$, then
\begin{equation}
    C(\portvec) \leq \xi.
\end{equation}
We will derive the price below, which is the ultimate goal of our analysis of quantum assets in this paper.

\subsubsection{Per-round win probabilities}\label{subsubsec:per-round-win-probs}
We now examine the per-round winning probabilities for each of the classical and quantum strategies given the market state $\rho(\eta)$ and input distribution $\pi(x,y)$, which we shorten to $\pi$ below. We start with the conditional classical per-round winning probability $\Pr_\mathrm{cl}[\mathrm{win}|x,y]$ of a classical strategy, and the maximal per-round winning probability
\begin{equation}
    p_\mathrm{win}^\mathrm{cl}(\pi) \coloneqq \max_{\text{cl. str.}}{\defsum{x,y}{}{\pi(x,y)\Pr_\mathrm{cl}[\mathrm{win}|x,y]}},
\end{equation}
where we are maximizing over classical strategies (which we denote as "cl. str.") satisfying the no-communication constraint. If we take $\pi(x,y)$ to be uniform, then $p_\mathrm{win}^\mathrm{cl}(\pi)$ reduces to the familiar $3/4$ success probability of the best classical strategy. Under a quantum strategy, given the market state $\rho(\eta)$, and defining
\begin{equation}
    \mathcal{M} = \parens*{\{\Pi_{a|x}^A\}_{a=0,1;x=0,1}, \{\Pi_{b|y}^B\}_{b=0,1;y=0,1}},
\end{equation}
as a family of local measurements, the maximal per-round winning probability is
\begin{equation}
    p_\mathrm{win}^\mathrm{q}(\eta,\pi) = \max_{\mathcal{M}}\defsum{x,y}{}{\pi(x,y)\defsum{a \oplus b = x \wedge y}{}{\Tr[\rho(\eta)\parens*{\Pi_{a|x}^A \otimes \Pi_{b|y}^B}]}}.
\end{equation}
Given the form of the quantum market state, which is a Bell state with probability $1-\eta$ and the maximally mixed state with probability $\eta$, this reduces to
\begin{equation}
    p_\mathrm{win}^\mathrm{q}(\eta,\pi) = (1-\eta)p_\mathrm{win}^\mathrm{q}(0,\pi) + \eta p_\mathrm{win}^\mathrm{rand},
\end{equation}
where $p_\mathrm{win}^\mathrm{rand}$ is the winning probability using a quantum strategy with the maximally mixed state. For the CHSH game, this value is $p_\mathrm{win}^\mathrm{rand} = 1/2$ for every allowed pair of rank-one local projective measurements, which is worse than the best classical strategy. Here, $p_\mathrm{win}^\mathrm{q}(0,\pi)$ is the winning probability using the optimal quantum strategy on the entangled quantum state.

Note that the maximal classical and quantum winning probabilities correspond to the graph invariant quantities of the exclusivity graph induced by the set of winning projectors $\{\Pi_{a|x}^A \otimes \Pi_{b|y}^B\}$. Given the set $\{\Pi_{a|x}^A \otimes \Pi_{b|y}^B\}$, the asset exclusivity graph $G$ was introduced in~\cref{subsec:exclusive-assets}. Then, given $G$, the best classical winning probability is given by the independence number $\alpha(G)$ of the graph:
\begin{equation}
    p_\mathrm{win}^\mathrm{cl} = \frac{\alpha(G)}{4},\label{eqn:classical-win-prob}
\end{equation}
where the normalization factor is equal to the size of the input space~\cite[Sec. 3.5]{scarpa2013quantum}. Likewise, the best quantum winning probability is given by the \lovasz theta number:
\begin{equation}
    p_\mathrm{win}^\mathrm{q}(0) = \frac{\vartheta(G)}{4}.
\end{equation}
Thus, the winning probability for a quantum strategy is
\begin{equation}
    p_\mathrm{win}^\mathrm{q}(\eta) = (1-\eta)\frac{\vartheta(G)}{4} + \frac{\eta}{2}.\label{eqn:quantum-win-prob}
\end{equation}
Given this, we define the \emph{quantum edge}:

\begin{definition}[Quantum edge per CHSH round]\label{def:quantum-edge}
    Let $\eta \in [0, 1]$ and $\pi \equiv \pi(x,y)$ be a probability distribution on inputs $x,y \in \{0,1\}$ to the CHSH market game. Assuming access to quantum resources, we define the \emph{quantum edge} per round as
    \begin{equation}
        \Delta(\eta,\pi) \coloneqq p_\mathrm{win}^\mathrm{q}(\eta,\pi) - p_\mathrm{win}^\mathrm{cl}(\pi).
    \end{equation}
\end{definition}

Note that there is a critical value for $\eta$, which we call $\eta_{\text{crit}}$, above which it is no longer beneficial to use a quantum strategy. We formalize this in the following theorem:

\begin{thm}[Critical probability]\label{thm:critical-probability}
    Given the quantum edge $\Delta(\eta, \pi)$ defined in~\cref{def:quantum-edge} and taking $\pi$ to be uniform over inputs $x,y \in \{0,1\}$ to the CHSH game, there is a critical value for $\eta$,
    \begin{equation}
        \eta_\mathrm{crit} = \frac{\vartheta(G) - \alpha(G)}{\vartheta(G) - 2} = 1 - \frac{1}{\sqrt{2}},
    \end{equation}
    where $G = \mathrm{Ci}_8(1,4)$, such that $\Delta(\eta,\pi) > 0$ if $\eta < \eta_\mathrm{crit}$, $\Delta(\eta,\pi) = 0$ at $\eta = \eta_\mathrm{crit}$, and $\Delta(\eta,\pi) < 0$ for $\eta > \eta_\mathrm{crit}$.
\end{thm}

\begin{proof}
    Recall the expressions for $p_\mathrm{win}^\mathrm{cl}$ in~\cref{eqn:classical-win-prob} and $p_\mathrm{win}^\mathrm{q}(\eta)$ in~\cref{eqn:quantum-win-prob}. The critical probability occurs when $p_\mathrm{win}^\mathrm{cl} = p_\mathrm{win}^\mathrm{q}(\eta)$, which is given by
    \begin{align}
        \frac{\alpha(G)}{4} &= (1-\eta)\frac{\vartheta(G)}{4} + \frac{\eta}{2}
        \implies \eta_{\text{crit}} = \frac{\vartheta(G) - \alpha(G)}{\vartheta(G) - 2}.
    \end{align}
    Substituting in the known values for $\vartheta(G)$ and $\alpha(G)$ then yields the value in the theorem statement. If $\eta > \eta_{\text{crit}}$, then $p_\mathrm{win}^\mathrm{cl} > p_\mathrm{win}^\mathrm{q}(\eta)$ and $\Delta(\eta, \pi) < 0$, thus losing the quantum edge. Analogously, if $\eta < \eta_\mathrm{crit}$, then $p_\mathrm{win}^\mathrm{cl} < p_\mathrm{win}^\mathrm{q}(\eta)$ and $\Delta(\eta,\pi) > 0$, and the quantum edge is retained. Finally, if $\eta = \eta_\mathrm{crit}$, then $p_\mathrm{win}^\mathrm{cl} = p_\mathrm{win}^\mathrm{q}(\eta)$ and thus $\Delta(\eta,\pi) = 0$.
\end{proof}

In the case that $p_\mathrm{win}^\mathrm{cl} > p_\mathrm{win}^\mathrm{q}$, it actually makes more sense to use a classical strategy, rather than a quantum one. The critical probability is not unique to our CHSH-based market and is a well-studied quantity in the contextuality literature~\cite{vertesi2008more,acin2006grothendiecks}, where the related quantity is called the "visibility."

\subsubsection{Multi-round analysis}\label{subsubsec:multi-round-analysis}
We now use the per-round winning probabilities to derive the bid and ask prices for Alice and Charlie, respectively, across the full $T$-round CHSH market game. We will see how starting with the exclusivity graph $G$, which we got from the set of winning projectors that make up the quantum assets in our CHSH market, and the noise parameter $\eta$ in the quantum market state, we can arrive at interesting dynamics between the traders that give rise to emergent financial properties like a bid-ask spread. Thus, from first principles, we derive a coherent pricing model for quantum assets in a quantum market.

Let $T \geq 1$ be the number of rounds of the CHSH game to be played according to the rules stated in the contract (see~\cref{subsec:contract-specification}). Let $S \in \{0, 1, \ldots, T\}$ be the realized number of rounds of the CHSH game that Alice wins. Again, we assume that the input distribution is uniform: $\pi(x,y) = \frac{1}{4}$ for all $(x,y)$; analogous results can be derived by allowing a more general input distribution. Before the $T$-round CHSH market game begins, at $t=0$, Alice purchases her portfolio of quantum assets from Charlie and pays the fee $C_0$ for a total cost of $C(\portvec)$. The portfolio is $\portvec = (\port_k)_{k \in \mathcal K}$ and contains projectors $\{\Pi_k\}$, where each projector costs $c_k$; we derive the price below. Thus, the total cost to buy the portfolio of quantum assets is $C_Q(\portvec) = \defsum{k}{}{c_k \port_k}$, as in~\cref{eq:portfolio_cost}. At maturity time $t = T$ (i.e., after all $T$ rounds of the CHSH game have been played), Charlie pays Alice the realized settlement $S$.

If no bundle remains for input pair $(x,y)$, the round uses the prescribed classical strategy. Otherwise, one unit of each asset indexed by $\mathcal{K}_{(x,y)}$ is available for the quantum round. We now analyze the expected number of rounds that Alice remains "quantum-solvent," that is, she has enough of the two winning projectors indexed by $\mathcal{K}_{(x,y)}$ to use a quantum strategy on input $(x,y)$. We define the random variable $Q_T$ to quantify the number of rounds in which Alice can use the quantum strategy:
\begin{equation}
    Q_T \coloneqq \defsum{t=1}{T}{\vb{1}\{\text{both projectors for $(x_t,y_t)$ are available at $t$ in $\portvec_t$}\}},
\end{equation}
where $\portvec_t$ is the portfolio at round $t$. Denote by $X_{(x,y)}(T)$ the number of times that input $(x,y)$ arrives in all $T$ rounds, and denote $W_{(x,y)}(\portvec)$ as the number of times both corresponding  projectors are available in the portfolio $\portvec$. That is, $W_{(x,y)}(\vb{w})$ is the initial number of bundles of the two winning projectors for input $(x,y)$. Hence, we have
\begin{equation}
    Q_T = \defsum{x,y \in \{0,1\}}{}{\min\{X_{(x,y)}(T), W_{(x,y)}(\portvec)\}}.
\end{equation}
We can model $X_{(x,y)}(T)$ as a binomial random variable. In each round, ${(x,y)}$ arrives independently with probability $1/4$. The total number of times that ${(x,y)}$ appears across the $T$ rounds is then
\begin{equation}
    X_{(x,y)}(T) \sim \text{Bin}(T, 1/4).
\end{equation}
Using this, we can find the expectation value of $Q_T$:
\begin{align}
    \EE[Q_T] &= \defsum{x,y \in \{0,1\}}{}{\EE\lb \min\{X_{(x,y)}(T), W_{(x,y)}(\vb{w})\}\rb}. 
\end{align}
We additionally assume equal initial bundle inventories $W_{(x,y)}(\vb{w}) = W$ for every input pair $(x,y)$. Likewise, the $X_{(x,y)}(T)$ are identically distributed\footnote{Note that the $X_{(x,y)}(T)$ are, however, correlated as if we receive a lot of input $(x,y)$, then we have fewer occurrences of the other inputs. Regardless, this does not affect our calculations.} for each $(x,y)$, so we take $X \equiv X_{(x,y)}(T)$. Thus, we have
\begin{align}
    \EE[Q_T] &= 4 \cdot \EE[\min\{X, W\}],
\end{align}
with $X \sim \text{Bin}(T, 1/4)$. Then, we calculate $\EE[\min\{X, W\}]$ as
\begin{align}
    \EE[\min\{X, W\}] &= \defsum{j=0}{W}{j\Pr[X=j]} + W\defsum{j=W+1}{T}{\Pr[X=j]}.
\end{align}
As $X \sim \text{Bin}(T, 1/4)$, we have
\begin{equation}
    \Pr[X=j] = \binom{T}{j}\lp \frac{1}{4} \rp^j\lp \frac{3}{4} \rp^{T-j},
\end{equation}
hence, we have
\begin{equation}
    \EE[Q_T] = 4 \cdot \lp \defsum{j=0}{W}{j\binom{T}{j}\lp \frac{1}{4} \rp^j\lp \frac{3}{4} \rp^{T-j}} + W\defsum{j=W+1}{T}{\binom{T}{j}\lp \frac{1}{4} \rp^j\lp \frac{3}{4} \rp^{T-j}} \rp.\label{eqn:expected-quantum-solvency}
\end{equation}
Now, we combine the expected number of quantum rounds with the winning probabilities derived in~\cref{subsubsec:per-round-win-probs}. Since the inputs are uniformly distributed and the initial inventory is the same for each input pair, every input pair contributes equally to the expected numbers of quantum and classical rounds. As such, the expected total payout is 
\begin{align}
    \EE[S \mid \eta] &= \EE_\pi[Q_T] \cdot p_\mathrm{win}^\mathrm{q}(\eta) + (T - \EE[Q_T]) \cdot p_\mathrm{win}^\mathrm{cl}\\
    &= Tp_\mathrm{win}^\mathrm{cl} + (\underbrace{p_\mathrm{win}^\mathrm{q}(\eta) - p_\mathrm{win}^\mathrm{cl}}_{\Delta(\eta)})\EE[Q_T],\label{eq:expectedpayout}
\end{align}
where $\EE[Q_T]$ is as given above, $p_\mathrm{win}^\mathrm{cl}$ and $p_\mathrm{win}^\mathrm{q}(\eta)$ are given in~\cref{eqn:classical-win-prob,eqn:quantum-win-prob}, respectively, and we highlight the difference between the probabilities as the quantum edge $\Delta(\eta)$ defined in~\cref{def:quantum-edge}.

We now calculate the \emph{net present value} (NPV), which takes into account the full cashflow for both sides of the trade. For Alice, this is the expected return given her estimate/belief for the value of $\eta$, which we denote as $\eta_A$, minus the upfront capital that she invests, which is the cost of the quantum assets. Thus, we have\footnote{Recall from~\cref{subsec:classical-assets} that, in general, a net risk-free return $r$ over the contract period would discount the expected settlement by $1+r$, but in this work we ignore the discounting rate and take $r=0$ to focus on the novel aspects of the quantum market dynamics.}
\begin{equation}
    \mathrm{NPV}_A = \EE[S \mid \eta_A] - C(\portvec).
\end{equation}
From Charlie's perspective, the net present value is
\begin{equation}
    \mathrm{NPV}_C = C(\portvec) - \EE[S \mid \eta_C],
\end{equation}
where $\eta_C$ is Charlie's estimate for the true value of $\eta$. In order for the trade to be worthwhile for either side, the NPV must be greater than or equal to 0. Thus, solving for $C(\vb{w})$ gives the maximum/minimum cost of the portfolio of assets that either side is willing to pay/receive. For Alice, the maximum amount that she is willing to pay for her portfolio of quantum assets is
\begin{equation}
    C_A^{\max}(\portvec) = \EE[S \mid \eta_A],
\end{equation}
which, in other words, is the \emph{bid price}. For Charlie, the minimum amount that he is willing to sell the set of quantum assets to Alice for is
\begin{equation}
    C_C^{\min}(\portvec) = \EE[S \mid \eta_C],
\end{equation}
which is the \emph{ask price}. A trade happens when both sides agree to a price $C(\portvec)$ such that
\begin{equation}
    C_C^{\min}(\portvec) \leq C(\portvec) \leq C_A^{\max}(\portvec),
\end{equation}
which forms the \emph{trade condition}. Assume for simplicity that the cost-per-asset $c$ is the same for all of the quantum assets. If Alice purchases $N > 0$ assets, then the realized price of each quantum asset is simply
\begin{equation}
    c = \frac{C_Q(\portvec)}{N},
\end{equation}
where, recall, we decomposed the total cost Alice pays to play the game as $C(\vb{w}) = C_0 + C_Q(\vb{w})$, where $C_0 = Tp_\mathrm{win}^\mathrm{cl}$ and where
\begin{equation}
    C_Q(\vb{w}) = \defsum{k \in \mathcal{K}}{}{w_kc_k}.
\end{equation}
In terms of bundles, this means that each two-asset bundle costs $2c$. This thus gives us a first-principles pricing mechanism for quantum assets when the market structure is based on a nonlocal quantum game.

Note that the catalyst for a spread in the bid and ask prices is disagreement over $\eta$, which appears in the market state $\rho(\eta)$. Recall from above that we considered three scenarios. In the first scenario, we assume that both Alice and Charlie know the true value of $\eta$ exactly. In the second, we assume they both estimate $\eta$ to be $\eta_A \neq \eta$, that is, they agree on an incorrect value for the noise. In the third scenario, Alice estimates $\eta$ to be $\eta_A$ and Charlie estimates it to be $\eta_C \neq \eta_A$. In the first and second scenarios, both Alice and Charlie agree on the expected winning probability and thus on the expected payoff, so they converge to the same price. They therefore both regard the trade at their common valuation as fair.

In the third scenario, there is belief disagreement with respect to the market state. Here, a natural bid-ask spread opens up from the differing fair prices that Alice and Charlie arrive at given their estimates $\eta_A$ and $\eta_C$, respectively, of the true value of $\eta$. This affects the bid and ask prices only in places where $\eta$ arises, namely, in the expression for the quantum winning probability $p_\mathrm{win}^\mathrm{q}(\eta)$. In particular, a trade only occurs when $\eta_A \leq \eta_C$, under the assumption that Alice is always the one going long on the trade. If, however, we change the market structure to allow either side to go long or short, based on their estimate for $\eta$, then trades can happen either when $\eta_A \leq \eta_C$ (Alice goes long, Charlie goes short) or when $\eta_A > \eta_C$ (Alice goes short, Charlie goes long). One way to realize Alice going short is that Alice and Charlie switch roles, where Charlie receives half of the quantum state and plays the CHSH game with Bob, though we leave further details to allowing shorting to future work. 

Such a difference in estimates can arise from several sources. One example is an asymmetry in the number of states that the traders have during the market research phase, as we alluded to earlier. That is, if we charge a trader for every copy of the market state that they use for estimating $\eta$, then if one side buys more copies, they will likely get a different (and better) estimate of $\eta$. At the same time, they now have less capital to dedicate towards the actual playing of the CHSH market game, so there is an interplay between information access and amount of capital that is able to be put up (which, in turn, decreases the expected return). We do not quantify this, though we leave it to future work to see how this tradeoff plays out.

\subsection{Numerical example}\label{subsec:numerical-examples}
We now give an illustrative numerical example demonstrating how our CHSH quantum market game works. Our goal with this example is to demonstrate clearly the cashflow, the details of how the bid and ask prices are obtained, the natural emergence of a bid-ask spread, and how changing the variables that go into the CHSH market game affects the resulting economics of the trade.

We start by stating some assumptions and simplifications that allow us to focus on the aforementioned features of the market. First, we assume that the market research has already been done by both parties, and each has obtained their estimate of $\eta$: Alice estimates $\eta_A$ and Charlie estimates $\eta_C \neq \eta_A$ (so we are in the third scenario of knowledge about $\eta$, where we have a belief disagreement). We do not account for the cost of buying copies of the market state for simplicity. Furthermore, we assume that Alice has purchased enough market research states to give her confidence that she has a good estimate; there is likely an interesting adaptive strategy that Alice can use that depends on her uncertainty of her estimate, though we leave this to future work. We also assume that the market state does not change throughout the game, that is, each copy of the market state has the form in~\cref{eq:noisy-market-state}, where $\eta$ is unknown to both parties. Recall that each quantum round consumes one matching two-asset bundle and one shared copy of the market state. We assign the following values, which are known to both parties and agreed upon per the financial contract:
\begin{itemize}
    \item $T = 800$ rounds of the CHSH game
    \item $N = 1600$ total winning projectors $\Pi_k$ (see~\cref{eq:def-quantum-assets}) are purchased by Alice (equivalently, 800 bundles, each containing the two winning projectors for each $(x,y)$)
    \item $\pi(x,y) = \frac{1}{4}$, so each input has an equal probability of appearing, so Alice buys $\port_k = 1600 \cdot \frac{1}{8} = 200$ of each of the 8 winning projectors in~\cref{eq:def-quantum-assets}. Equivalently, she buys 200 of each bundle corresponding to input $(x,y)$, each of which has 200 of each winning asset.
\end{itemize}
Assume that Alice estimates $\eta$ to be $\eta_A = 0.10$, while Charlie estimates it to be $\eta_C = 0.15$. Thus, Charlie believes the market state is \emph{less} "quantum" than Alice believes it to be. Since we only care about computing bid and ask prices to ascertain whether or not a trade will occur, the true value of $\eta$ does not matter. If we were to analyze the actual outcome of playing the game, then it would matter.

We know that the success probabilities under the classical and prescribed quantum strategies are $p_\mathrm{win}^\mathrm{cl} = 0.75$ and
\begin{equation}
    p_\mathrm{win}^\mathrm{q}(\eta) = \frac{2+\sqrt{2}}{4} \cdot (1-\eta) + \eta \cdot p_\mathrm{win}^\mathrm{rand} \approx 0.853553 \cdot (1-\eta) + \frac{\eta}{2},
\end{equation}
respectively, so Alice finds
\begin{equation}
    p_\mathrm{win}^\mathrm{q}(\eta_A) \approx 0.818198
\end{equation}
and Charlie finds
\begin{equation}
    p_\mathrm{win}^\mathrm{q}(\eta_C) \approx 0.80052.
\end{equation}
The expected payout is then given by~\cref{eq:expectedpayout}. Using the expression in~\cref{eqn:expected-quantum-solvency}, we find $\EE[Q_T] = 780.465$, so substituting in the appropriate values for $\eta_A$ and $\eta_C$, we find
\begin{align}
    \EE[S \mid \eta_A] &\approx 653.23\\
    \EE[S \mid \eta_C] &\approx 639.43.
\end{align}
Now, we calculate the net present value for each party, taking into account the full cashflow of the trade. For Alice, her net present value is
\begin{equation}
    \mathrm{NPV}_A = \EE[S \mid \eta_A] - C(\portvec),
\end{equation}
and Charlie's net present value is
\begin{equation}
    \mathrm{NPV}_C = C(\portvec) - \EE[S \mid \eta_C].
\end{equation}
We can now determine the bid and ask prices. Alice participates in a trade if and only if $\mathrm{NPV}_A \geq 0$, so
\begin{equation}
    C_A^{\max}(\portvec) = \EE[S \mid \eta_A].
\end{equation}
$C_A^{\max}(\portvec)$ is the maximum total portfolio cost that Alice is willing to pay to Charlie, which, in other words, is the maximum ask price Alice would be willing to accept. Likewise, for Charlie, he will participate in the trade if and only if $\mathrm{NPV}_C \geq 0$ as well, so we have
\begin{equation}
    C_C^{\min}(\portvec) = \EE[S \mid \eta_C],
\end{equation}
where $C_C^{\min}(\portvec)$ is the minimum total portfolio cost he will accept from Alice to participate in the trade, which is the minimum bid price he would be willing to accept. Computing each of these values using the numerical values we have found, we get
\begin{align}
    C_A^{\max}(\portvec) &\approx 653.23,\\
    C_C^{\min}(\portvec) &\approx 639.43.
\end{align}
We decompose these into the classical premium $C_0 = Tp_\mathrm{win}^{\mathrm{cl}} = 800 \times 0.75 = 600$ and the "quantum upgrade" $C_Q(\vb{w})$. Subtracting $C_0 = 600$ from the two bounds above, we find the bid-ask spread for the quantum assets to be 39.43--53.23, so the fair quantum upgrade cost $C_Q(\vb{w})$ for the trade lies in the interval
\begin{equation}
    39.43 \leq C_Q(\vb{w}) \leq 53.23.
\end{equation}
A natural choice would be the midpoint of the bid and ask prices:
\begin{equation}
    C_Q(\vb{w}) = \frac{39.43+53.23}{2} = 46.33,
\end{equation}
or
\begin{equation}
    c = \frac{46.33}{1600} \approx 0.0290
\end{equation}
per asset, or $2c \approx 0.0579$ per bundle. This trade is driven by disagreement over the true value of $\eta$, showing how belief disagreement in such a market can lead to financially rational behavior from both parties.

\section{Discussion}\label{sec:discussion}
In this work, we introduced a quantum market based on nonlocal games and developed a graph-theoretic framework for pricing exclusive quantum assets. Within this framework, classical and quantum strategies are governed by different graph-invariant quantities, leading to a quantum-classical separation in attainable payoffs. We then realized this framework in a CHSH-based market with a trusted exchange, noisy Werner states describing the state of the market, and tradable objects based on projective measurements. This yields a first-principles pricing mechanism that naturally leads to a bid-ask spread through asymmetric beliefs about the market state. By connecting our ideas to previous work by Bao and Rebentrost~\cite{bao2023fundamental} and Hughston and S\'anchez-Betancourt~\cite{hughston2024valuation}, we hope that this work furthers the development of research on  potential future markets populated by quantum assets.

One may ask what the motivation is for considering such a market. One motivation, offered by Bao and Rebentrost~\cite{bao2023fundamental}, is that if quantum devices are so ubiquitous in the future, then situations may arise where financial decisions are made based on measuring quantum states. As such, it would be natural for a quantum market to develop, and it is imperative to understand how such a market functions. That is, if quantum computers are routinely used to efficiently solve \emph{classical} financial problems, then it is feasible that this frequent use of quantum computation could inspire the development of financial assets or instruments that are purely \emph{quantum}. Having a framework for understanding this potential landscape and some of its basic properties is thus beneficial.

Future work could consider other types of quantum games, such as the quantum Prisoners' Dilemma~\cite{eisert1999quantum} and the other types of nonlocal games that we mentioned in~\cref{sec:chsh-market-game} (e.g., the $I_{3322}$ game), extending these frameworks to a financial setting. We also leave to future work some of the finer details in the analysis that we mentioned throughout the text, such as the cost of the market research phase and the effects of uncertainty in estimates for $\eta$ from each trader's perspective. Other additions include changing conditions, such as time-varying market states, or a changing input distribution, which would require adaptive strategies and a richer landscape of learning states or probability distributions. If the value of the quantum assets also changed over time, concepts like discounting by risk-free assets or arbitrage may become relevant as well.

We further leave to future work a quantum market based on the \emph{compatibility} framework, which is in contrast to the exclusivity graph framework considered in this paper. In the compatibility framework, we consider only \emph{compatible} measurements, that is, measurements that commute and, therefore, can be performed at the same time. In this setting, the fundamental object is the \emph{compatibility hypergraph}~\cite[Ch. 2]{amaral2018graph}, which describes the set of measurements that can be done simultaneously, in contrast to the exclusivity graph, which organizes mutually exclusive events. In fact, the exclusivity graph can be derived from the compatibility hypergraph, so it would be interesting to see what generalizations and implications for a quantum financial model one would be able to extract using the compatibility framework.

Finally, we mention the potential inclusion of one or more quantum cryptographic primitives that may give rise to interesting financial phenomena. For example, quantum money~\cite{wiesner1983cojugate,aaronson2009quantum,aaronson2012quantum,farhi2012quantum,pastawski2012unforgeable,mosca2009quantum} schemes could be included in some way, adapting the projectors appropriately to be represented as quantum money states. Then, the market that we introduced inherits the benefits that come with quantum money, most notably the resistance to forgery, which would prevent any market participants from spending assets that they do not actually own. Other primitives that could be useful include bit commitment schemes or digital signatures, each of which can be implemented using, for example, pseudorandom states~\cite{ji2018pseudorandom,ananth2022pseudorandom,ananth2022cryptography,morimae2022quantum,coladangelo2024black-box}. Such primitives would allow for stronger security and a more robust market, where trust is an important ingredient for a properly functioning financial system.

\section*{Acknowledgments}
D.J.S.~acknowledges funding support from a Graduate Research Fellowship from the Joint Quantum Institute (JQI) at the University of Maryland, College Park. B.C., D.-L.V., and P.R.~are supported by the National Research Foundation, Singapore through the National Quantum Office, hosted in A*STAR, under its Centre for Quantum Technologies Funding Initiative and its Advanced Quantum Algorithms and Solutions Funding Initiative. K.B.~is supported by a Hartree fellowship from the Joint Center for Quantum Information and Computer Science (QuICS) at the University of Maryland, College Park. D.J.S.~and A.V.G.~acknowledge support from the U.S.~Department of Energy, Office of Science, Accelerated Research in Quantum Computing, Fundamental Algorithmic Research toward Quantum Utility (FAR-Qu). D.J.S.~and A.V.G.~were also supported in part by the DoE ASCR Quantum Testbed Pathfinder program (award No.~DE-SC0024220), ONR MURI, NSF QLCI (awards No.~OMA-2120757 and No.~OSI-2553574), NSF STAQ program, AFOSR MURI, ARL (W911NF-24-2-0107), and NQVL:QSTD:Design:FTL. D.J.S.~and A.V.G.~also acknowledge support from the U.S.~Department of Energy, Office of Science, National Quantum Information Science Research Centers, Quantum Systems Accelerator (award No.~DE-SCL0000121).

\printbibliography

@article{alcazar2022quantum,
    title={Quantum algorithm for credit valuation adjustments},
    author={Javier Alcazar and Andrea Cadarso and Amara Katabarwa and Marta Mauri and Borja Peropadre},
    journal={New Journal of Physics},
    volume={24},
    pages={023036},
    year={2022},
    doi={10.1088/1367-2630/ac5003}
}

@book{amaral2018graph,
    title={On Graph Approaches to Contextuality and their Role in Quantum Theory},
    author={Barbara Amaral and Marcelo Terra Cunha},
    publisher={Springer},
    year={2018},
    url={https://link.springer.com/book/10.1007/978-3-319-93827-1}
}

@article{an2021quantum-accelerated,
    title={Quantum-accelerated multilevel Monte Carlo methods for stochastic differential equations in mathematical finance},
    author={Dong An and Noah Linden and Jin-Peng Liu and Ashley Montanaro and Changpeng Shao and Jiasu Wang},
    journal={Quantum},
    volume={5},
    pages={481},
    year={2021},
    doi={10.22331/q-2021-06-24-481}
}

@article{bao2023fundamental,
    title={Fundamental theorem for the pricing of quantum assets},
    author={Jinge Bao and Patrick Rebentrost},
    journal={arXiv:2212.13815},
    year={2023},
    doi={10.48550/arXiv.2212.13815}
}

@book{bell2004speakable,
    title={Speakable and Unspeakable in Quantum Mechanics: Collected Papers on Quantum Philosophy},
    author={Bell, J. S. and Aspect, Alain},
    ISBN={9780511815676},
    doi={10.1017/cbo9780511815676},
    publisher={Cambridge University Press},
    year={2004}
}

@article{bharti2022graph-theoretic,
    title={Graph-theoretic approach for self-testing in Bell scenarios},
    author={Kishor Bharti and Maharshi Ray and Zhen-Peng Xu and Masahito Hayashi and Leong-Chuan Kwek and Ad\'{a}n Cabello},
    journal={PRX Quantum},
    volume={3},
    pages={030344},
    year={2022},
    doi={10.1103/PRXQuantum.3.030344}
}

@article{bouland2020prospects,
    author={Adam Bouland and Wim van Dam and Hamed Joorati and Iordanis Kerenidis and Anupam Prakash},
    title={Prospects and challenges of quantum finance},
    journal={arXiv:2011.06492v1},
    year={2020},
    doi={10.48550/arXiv.2011.06492}
}

@article{brunner2014bell,
    title={Bell nonlocality},
    author={Nicolas Brunner and Daniel Cavalcanti and Stefano Pironio and Valerio Scarani and Stephanie Wehner},
    journal={Reviews of Modern Physics},
    volume={86},
    issue={2},
    pages={419},
    year={2014},
    doi={10.1103/RevModPhys.86.419}
}

@article{cabello2014graph,
    title={Graph-theoretic approach to quantum correlations},
    author={Cabello, Ad\'{a}n and Severini, Simone and Winter, Andreas},
    journal={Physical Review Letters},
    volume={112},
    number={4},
    pages={040401},
    year={2014},
    doi={10.1103/PhysRevLett.112.040401},
}

@article{cabello2010noncontextuality,
    title={(Non)-contextuality of physical theories as an axiom},
    author={Adan Cabello and Simone Severini and Andreas Winter},
    journal={arXiv:1010.2163},
    year={2010},
    doi={10.48550/arXiv.1010.2163}
}

@article{cabello2013basic,
    title={Basic exclusivity graphs in quantum correlations},
    author={Ad\'{a}n Cabello and Lars Eirik Danielsen and Antonio J. L\'{o}pez-Tarrida and Jos\'{e} R. Portillo},
    journal={Physical Review A},
    volume={88},
    pages={032104},
    year={2013},
    doi={10.1103/PhysRevA.88.032104}
}

@article{chakrabarti2021threshold,
    author={Shouvanik Chakrabarti and Rajiv Krishnakumar and Guglielmo Mazzola and Nikitas Stamatopoulos and Stefan Woerner and William J. Zeng},
    title={A Threshold for Quantum Advantage in Derivative Pricing},
    journal={Quantum},
    volume={5},
    pages={463},
    year={2021},
    doi={10.22331/q-2021-06-01-463}
}

@article{clauser1969proposed,
    title={Proposed experiment to test local hidden-variable theories},
    author={Clauser, John F. and Horne, Michael A. and Shimony, Abner and Holt, Richard A.},
    journal={Physical Review Letters},
    volume={23},
    issue={15},
    pages={880--884},
    year={1969},
    doi={10.1103/PhysRevLett.23.880}
}

@article{dalzell2022end-to-end,
    title={End-to-end resource analysis for quantum interior point methods and portfolio optimization},
    author={Alexander M. Dalzell and B. David Clader and Grant Salton and Mario Berta and Cedric Yen-Yu Lin and David A. Bader and Nikitas Stamatopoulos and Martin J.A. Schuetz and Fernando G.S.L. Brand{\~a}o and Helmut G. Katzgraber and William J. Zeng},
    journal={PRX Quantum},
    volume={4},
    pages={040325},
    year={2023},
    doi={10.1103/PRXQuantum.4.040325}
}

@article{doriguello2022quantum,
    title={Quantum algorithm for stochastic optimal stopping problems with applications in finance},
    author={Jo{\~{a}}o F. Doriguello and Alessandro Luongo and Jinge Bao and Patrick Rebentrost and Miklos Santha},
    journal={17\textsuperscript{th} Conference on the Theory of Quantum Computation, Communication and Cryptography},
    pages={1--24},
    year={2022},
    doi={10.4230/LIPIcs.TQC.2022.2}
}

@article{egger2021credit,
    title={Credit risk analysis using quantum computers},
    author={Daniel J. Egger and Ricardo Garc\'{i}a Guti\'{e}rrez and Jordi Cahu\'{e} Mestre and Stefan Woerner},
    journal={IEEE Transactions on Computers},
    volume={70},
    number={12},
    year={2021},
    doi={10.1109/TC.2020.3038063}
}

@article{einstein1935can,
    title={Can quantum-mechanical description of physical reality be considered complete?},
    author={Einstein, A. and Podolsky, B. and Rosen, N.},
    journal={Physical Review},
    volume={47},
    issue={10},
    pages={777--780},
    year={1935},
    doi={10.1103/PhysRev.47.777},
}

@book{grotschel2012geometric,
    title={Geometric algorithms and combinatorial optimization},
    author={Gr{\"o}tschel, Martin and Lov{\'a}sz, L{\'a}szl{\'o} and Schrijver, Alexander},
    volume={2},
    year={2012},
    publisher={Springer Science \& Business Media},
    url={https://link.springer.com/book/10.1007/978-3-642-78240-4}
}

@article{herman2023quantum,
    title={Quantum computing for finance},
    author={Dylan Herman and Cody Googin and Xiaoyuan Liu and Yue Sun and Alexey Galda and Ilya Safro and Marco Pistoia and Yuri Alexeev},
    journal={Nature Reviews Physics},
    volume={5},
    pages={450--465},
    year={2023},
    doi={10.1038/s42254-023-00603-1}
}

@article{herman2023constrained,
    title={Constrained optimization via quantum Zeno dynamics},
    author={Dylan Herman and Ruslan Shaydulin and Yue Sun and Shouvanik Chakrabarti and Shaohan Hu and Pierre Minssen and Arthur Rattew and Romina Yalovetzky and Marco Pistoia},
    journal={Communications Physics},
    volume={6},
    number={219},
    year={2023},
    doi={10.1038/s42005-023-01331-9}
}

@article{hughston2024valuation,
    title={Valuation of a financial claim contingent on the outcome of a quantum measurement},
    author={Lane P. Hughston and Leandro S\'{a}nchez-Betancourt},
    journal={Journal of Physics A: Mathematical and Theoretical},
    volume={57},
    number={28},
    pages={285302},
    year={2024},
    doi={10.1088/1751-8121/ad4cab}
}

@article{kerenidis2019quantum,
    title={Quantum algorithms for portfolio optimization},
    author={Iordanis Kerenidis and Anupam Prakash and D\'{a}niel Szil\'{a}gyi},
    journal={AFT '19: Proceedings of the 1st ACM Conference on Advances in Financial Technologies},
    pages={147--155},
    year={2019},
    doi={10.1145/3318041.3355465}
}

@article{khan2025quantum,
    title={Quantum advantage in trading: A game theoretic approach},
    author={Faisal Shah Khan and Norbert M. Linke and Anton Trong Than and Dror Baron},
    journal={Quantum Economics and Finance},
    volume={2},
    issue={1},
    year={2025},
    doi={10.1177/29767032251333418}
}

@article{knuth1993sandwich,
    title={The sandwich theorem},
    author={Donald E. Knuth},
    journal={The Electronic Journal of Combinatorics},
    volume={1},
    year={1994},
    doi={10.37236/1193}
}

@article{leclerc2024financial,
    title={Financial risk management on a neutral atom quantum processor},
    author={Leclerc, Lucas and Ortiz-Guti\'errez, Luis and Grijalva, Sebasti\'an and Albrecht, Boris and Cline, Julia R. K. and Elfving, Vincent E. and Signoles, Adrien and Henriet, Lo\"{i}c and Del Bimbo, Gianni and Sheikh, Usman Ayub and Shah, Maitree and Andrea, Luc and Ishtiaq, Faysal and Duarte, Andoni and Mugel, Sam and C\'aceres, Irene and Kurek, Michel and Or\'us, Roman and Seddik, Achraf and Hammami, Oumaima and Isselnane, Hacene and M'tamon, Didier},
    journal={Physical Review Research},
    volume={5},
    pages={043117},
    year={2023},
    doi={10.1103/PhysRevResearch.5.043117}
}

@article{lim2024quantum,
    title={A quantum online portfolio optimization algorithm},
    author={Debbie Lim and Patrick Rebentrost},
    journal={Quantum Information Processing},
    volume={23},
    number={63},
    year={2024},
    doi={10.1007/s11128-024-04256-6}
}

@article{lovasz1979shannon,
    title={On the Shannon capacity of a graph},
    author={L. Lov\'{a}sz},
    journal={IEEE Transactions on Information Theory},
    volume={25},
    issue={1},
    pages={1--7},
    year={1979},
    doi={10.1109/TIT.1979.1055985}
}

@article{melnyk2008quantum,
    title={Quantum analog of the Black-Scholes formula (market of financial derivatives as a continuous weak measurements)},
    author={S.I. Melnyk and I.G. Tuluzov},
    journal={Electronic Journal of Theoretical Physics (EJTP)},
    volume={5},
    issue={18},
    pages={95--104},
    year={2008},
    url={https://philarchive.org/rec/MELQAO}
}

@article{naik2025portfolio,
    title={From portfolio optimization to quantum blockchain and security: a systematic review of quantum computing in finance},
    author={Abha Naik and Esra Yeniaras and Gerhard Hellstern and Grishma Prasad and Sanjay Kumar and Lalta Prasad Vishwakarma},
    journal={Financial Innovation},
    volume={11},
    pages={88},
    year={2025},
    doi={10.1186/s40854-025-00751-6}
}

@article{orus2019quantum,
    title={Quantum computing for finance: Overview and prospects},
    author={Rom\'{a}n Or\'{u}s and Samuel Mugel and Enrique Lizaso},
    journal={Reviews in Physics},
    volume={4},
    pages={100028},
    year={2019},
    doi={10.1016/j.revip.2019.100028}
}

@article{schaden2002quantum,
    title={Quantum finance},
    author={Martin Schaden},
    journal={Physica A: Statistical Mechanics and its Applications},
    volume={316},
    issue={1--4},
    pages={511--538},
    year={2002},
    doi={10.1016/S0378-4371(02)01200-1}
}

@article{segal1998black-scholes,
    title={The Black-Scholes pricing formula in the quantum context},
    author={Wiliam Segal and I.E. Segal},
    journal={Proceedings of the National Academy of Sciences},
    volume={95},
    issue={7},
    pages={4072--4075},
    year={1998},
    doi={10.1073/pnas.95.7.4072}
}

@article{stamatopoulos2020option,
    title={Option pricing using quantum computers},
    author={Nikitas Stamatopoulos and Daniel J. Egger and Yue Sun and Christa Zoufal and Raban Iten and Ning Shen and Stefan Woerner},
    journal={Quantum},
    volume={4},
    pages={291},
    year={2020},
    doi={10.22331/q-2020-07-06-291}
}

@article{stamatopoulos2024derivative,
    title={Derivative pricing using quantum signal processing},
    author={Nikitas Stamatopoulos and William J. Zeng},
    journal={Quantum},
    volume={8},
    pages={1322},
    year={2024},
    doi={10.22331/q-2024-04-30-1322}
}

@article{supic2020self-testing,
    title={Self-testing of quantum systems: a review},
    author={Ivan \v{S}upi\'{c} and Joseph Bowles},
    journal={Quantum},
    volume={4},
    pages={337},
    year={2020},
    doi={10.22331/q-2020-09-30-337}
}

@article{varian1987arbitrage,
    title={The arbitrage principle in financial economics},
    author={Hal R. Varian},
    journal={Journal of Economic Perspectives},
    volume={1},
    number={2},
    pages={55--72},
    year={1987},
    doi={10.1257/jep.1.2.55}
}

@article{wang2024option,
    title={Option pricing under stochastic volatility on a quantum computer},
    author={Guoming Wang and Angus Kan},
    journal={Quantum},
    volume={8},
    pages={1504},
    year={2024},
    doi={10.22331/q-2024-10-23-1504}
}

@article{woerner2019quantum,
    title={Quantum risk analysis},
    author={Stefan Woerner and Daniel J. Egger},
    journal={npj Quantum Information},
    volume={5},
    number={15},
    year={2019},
    doi={10.1038/s41534-019-0130-6}
}

@article{yalovetzky2023solving,
    title={Solving linear systems on quantum hardware with hybrid HHL++},
    author={Romina Yalovetzky and Pierre Minssen and Dylan Herman and Marco Pistoia},
    journal={Scientific Reports},
    volume={14},
    pages={20610},
    year={2024},
    doi={10.1038/s41598-024-69077-0}
}

@article{arrow1954existence,
    title={Existence of an equilibrium for a competitive economy},
    author={Kenneth J. Arrow and Gerard Debreu},
    journal={Econometrica},
    volume={22},
    number={3},
    pages={265--290},
    year={1954},
    doi={10.2307/1907353}
}

@phdthesis{scarpa2013quantum,
    author={Giannicola Scarpa},
    title={{Quantum Entanglement in Non-local Games, Graph Parameters and Zero-error Information Theory}},
    school={Universiteit van Amsterdam},
    year={2013},
    url={https://eprints.illc.uva.nl/id/eprint/2121/1/DS-2013-03.text.pdf}
}

@article{eisert1999quantum,
    title={Quantum games and quantum strategies},
    author={Jens Eisert and Martin Wilkens and Maciej Lewenstein},
    journal={Physical Review Letters},
    volume={83},
    pages={3077},
    year={1999},
    doi={10.1103/PhysRevLett.83.3077}
}

@article{juhasz1982asymptotic,
    title={The asymptotic behavior of Lov\'{a}sz $\vartheta$ function for random graphs},
    author={Ferenc Juh\'{a}sz},
    journal={Combinatorica},
    volume={2},
    number={2},
    pages={153--155},
    year={1982},
    url={https://real.mtak.hu/35583/8/Juhasz%2C%20F.%20-%20The%20asymptotic%20behaviour%20of%20Lovasz%27%20theta%20function%20for%20random%20graphs.pdf}
}

@article{bollobas1976cliques,
    title={Cliques in random graphs},
    author={B. Bollobas and P. Erd\H{o}s},
    journal={Mathematical Proceedings of the Cambridge Philosophical Society},
    volume={80},
    issue={3},
    pages={419--427},
    year={1976},
    doi={10.1017/S0305004100053056}
}

@article{grimmett1975colouring,
    title={On colouring random graphs},
    author={G.R. Grimmett and C.J.H. McDiarmid},
    journal={Mathematical Proceedings of the Cambridge Philosophical Society},
    volume={77},
    issue={2},
    pages={313--324},
    year={1975},
    doi={10.1017/S0305004100051124}
}

@article{frieze1990independence,
    title={On the independence number of random graphs},
    author={A.M. Frieze},
    journal={Discrete Mathematics},
    volume={81},
    issue={2},
    pages={171--175},
    year={1990},
    doi={10.1016/0012-365X(90)90149-C}
}

@article{matula1970complete,
    title={On the complete subgraphs of a random graph},
    author={D.W. Matula},
    journal={Combinatory Mathematics and its Applications},
    pages={356--369},
    year={1970},
    url={https://s2.smu.edu/~matula/Tech-Report70.pdf}
}

@article{hymas2026quenching,
    title={Quenching speculation in quantum markets via entangled neural traders},
    author={Kieran Hymas and Hiu Ming Lau and Kareem Raslan and Qiang Sun and Azhar Iqbal and Derek Abbott and Andrew D. Greentree and James Q. Quach},
    journal={arXiv:2602.06367},
    year={2026},
    doi={10.48550/arXiv.2602.06367}
}

@article{pavicic2023quantum,
    title={Quantum contextuality},
    author={Mladen Pavicic},
    journal={Quantum},
    volume={7},
    pages={953},
    year={2023},
    doi={10.22331/q-2023-03-17-953}
}

@article{budroni2022kochen-specker,
    title={Kochen-Specker contextuality},
    author={Costantino Budroni and Ad\'{a}n Cabello and Otfried G\"{u}hne and Matthias Kleinmann and Jan-{\AA}ke Larsson},
    journal={Reviews of Modern Physics},
    volume={94},
    pages={045007},
    year={2022},
    doi={10.1103/RevModPhys.94.045007}
}

@article{rebentrost2024quantum,
    title={Quantum Computational Finance: Quantum Algorithm for Portfolio Optimization},
    author={Patrick Rebentrost and Seth Lloyd},
    journal={KI -- K\"{u}nstliche Intelligenz},
    volume={38},
    pages={327--338},
    year={2024},
    doi={10.1007/s13218-024-00870-9}
}

@article{rebentrost2018quantum,
    title={Quantum computational finance: Monte Carlo pricing of financial derivatives},
    author={Patrick Rebentrost and Brajesh Gupta and Thomas R. Bromley},
    journal={Physical Review A},
    volume={98},
    pages={022321},
    year={2018},
    doi={10.1103/PhysRevA.98.022321}
}

@article{herman2026quantum,
    title={Quantum Speedups for Derivative Pricing Beyond Black-Scholes},
    author={Dylan Herman and Yue Sun and Jin-Peng Liu and Marco Pistoia and Charlie Che and Rob Otter and Shouvanik Chakrabarti and Aram Harrow},
    journal={arXiv:2602.03725},
    year={2026},
    doi={10.48550/arXiv.2602.03725}
}

@article{coladangelo2024black-box,
    title={On black-box separations of quantum digital signatures from pseudorandom states},
    author={Andrea Coladangelo and Saachi Mutreja},
    journal={Theory of Cryptography},
    volume={15366},
    pages={289--317},
    year={2024},
    doi={10.1007/978-3-031-78020-2_10}
}

@article{ananth2022pseudorandom,
    title={Pseudorandom (function-like) quantum state generators: new definitions and applications},
    author={Prabhanjan Ananth and Aditya Gulati and Luowen Qian and Henry Yuen},
    journal={Theory of Cryptography},
    pages={237--265},
    year={2022},
    doi={10.1007/978-3-031-22318-1_9}
}

@article{ananth2022cryptography,
    title={Cryptography from pseudorandom quantum states},
    author={Prabhanjan Ananth and Luowen Qian and Henry Yuen},
    journal={Advances in Cryptology---CRYPTO 2022},
    volume={13507},
    pages={208--236},
    year={2022},
    url={10.1007/978-3-031-15802-5_8}
}

@article{ji2018pseudorandom,
    title={Pseudorandom quantum states},
    author={Zhengfeng Ji and Yi-Kai Liu and Fang Song},
    journal={Advances in Cryptology -- CRYPTO 2018},
    pages={126--152},
    year={2018},
    doi={10.1007/978-3-319-96878-0_5}
}

@article{morimae2022quantum,
    title={Quantum commitments and signatures without one-way functions},
    author={Tomoyuki Morimae and Takashi Yamakawa},
    journal={Advances in Cryptology -- CRYPTO 2022},
    pages={269--295},
    year={2022},
    doi={10.1007/978-3-031-15802-5_10}
}

@article{wiesner1983cojugate,
    title={Conjugate coding},
    author={Stephen Wiesner},
    journal={ACM SIGACT News},
    volume={15},
    issue={1},
    pages={78--88},
    year={1983},
    doi={10.1145/1008908.1008920}
}

@article{aaronson2012quantum,
    title={Quantum money from hidden subspaces},
    author={Scott Aaronson and Paul Christiano},
    journal={STOC '12: Proceedings of the forty-fourth annual ACM symposium on Theory of computing},
    pages={41--60},
    year={2012},
    doi={10.1145/2213977.2213983}
}

@article{aaronson2009quantum,
    title={Quantum copy-protection and quantum money},
    author={Scott Aaronson},
    journal={2009 24th Annual IEEE Conference on Computational Complexity},
    year={2009},
    doi={10.1109/CCC.2009.42}
}

@article{farhi2012quantum,
    title={Quantum money from knots},
    author={Edward Farhi and David Gosset and Avinatan Hassidim and Andrew Lutomirski and Peter Shor},
    journal={ITCS '12: Proceedings of the 3rd Innovations in Theoretical Computer Science Conference},
    pages={276--289},
    year={2012},
    doi={10.1145/2090236.2090260}
}

@article{pastawski2012unforgeable,
    title={Unforgeable noise-tolerant quantum tokens},
    author={Fernando Pastawski and Norman Y. Yao and Liang Jiang and Mikhail D. Lukin and J. Ignacio Cirac},
    journal={Proceedings of the National Academy of Sciences},
    volume={109},
    issue={40},
    pages={16079--16082},
    year={2012},
    doi={10.1073/pnas.1203552109}
}

@article{mosca2009quantum,
    title={Quantum coins},
    author={Michele Mosca and Douglas Stebila},
    journal={Error-Correcting Codes, Finite Geometries and Cryptography. Contemporary Mathematics},
    volume={523},
    pages={35--47},
    year={2010},
    doi={10.1090/conm/523}
}

@article{cirelson1980quantum,
    title={Quantum generalizations of Bell's inequality},
    author={B.S. Cirel'son},
    journal={Letters in Mathematical Physics},
    volume={4},
    pages={93--100},
    year={1980},
    doi={10.1007/BF00417500}
}

@article{cleve2004consequences,
    title={Consequences and limits of nonlocal strategies},
    author={Richard Cleve and Peter Hoyer and Benjamin Toner and John Watrous},
    journal={CCC '04: Proceedings of the 19th IEEE Annual Conference on Computational Complexity},
    pages={236--249},
    year={2004},
    doi={10.1109/CCC.2004.1313847}
}

@article{briet2009generalized,
    title={A generalized Grothendieck inequality and entanglement in XOR games},
    author={Jop Bri\"{e}t and Harry Buhrman and Ben Toner},
    journal={arXiv:0901.2009},
    year={2009},
    doi={10.48550/arXiv.0901.2009}
}

@article{collins2004relevant,
    title={A relevant two qubit Bell inequality inequivalent to the CHSH inequality},
    author={Daniel Collins and Nicolas Gisin},
    journal={Journal of Physics A: Mathematical and General},
    volume={37},
    number={5},
    pages={1775},
    year={2004},
    doi={10.1088/0305-4470/37/5/021}
}

@article{svetlichny1987distinguishing,
    title={Distinguishing three-body from two-body nonseparability by a Bell-type inequality},
    author={George Svetlichny},
    journal={Physical Review D},
    volume={35},
    pages={3066},
    year={1987},
    doi={10.1103/PhysRevD.35.3066}
}

@article{horodecki1995violating,
    title={Violating Bell inequality by mixed spin-12 states: necessary and sufficient condition},
    author={R. Horodecki and P. Horodecki and M. Horodecki},
    journal={Physics Letters A},
    volume={200},
    issue={5},
    pages={340--344},
    year={1995},
    doi={10.1016/0375-9601(95)00214-N}
}

@article{werner1989quantum,
    title={Quantum states with Einstein-Podolsky-Rosen correlations admitting a hidden-variable model},
    author={Reinhard F. Werner},
    journal={Physical Review A},
    volume={40},
    pages={4277},
    year={1989},
    doi={10.1103/PhysRevA.40.4277}
}

@article{barrett2002nonsequential,
    title={Nonsequential positive-operator-valued measurements on entangled mixed states do not always violate a Bell inequality},
    author={Jonathan Barrett},
    journal={Physical Review A},
    volume={65},
    pages={042302},
    year={2002},
    doi={10.1103/PhysRevA.65.042302}
}

@article{acin2006grothendiecks,
    title={Grothendieck's constant and local models for noisy entangled quantum states},
    author={Antonio Ac\'{i}n and Nicolas Gisin and Benjamin Toner},
    journal={Physical Review A},
    volume={73},
    pages={062105},
    year={2006},
    doi={10.1103/PhysRevA.73.062105}
}

@article{elitzur1992quantum,
    title={Quantum nonlocality for each pair in an ensemble},
    author={Avshalom C. Elitzur and Sandu Popescu and Daniel Rohrlich},
    journal={Physics Letters A},
    volume={162},
    issue={1},
    pages={25--28},
    year={1992},
    doi={10.1016/0375-9601(92)90952-I}
}

@article{abramsky2017contextual,
    title={Contextual fraction as a measure of contextuality},
    author={Samson Abramsky and Rui Soares Barbosa and Shane Mansfield},
    journal={Physical Review Letters},
    volume={119},
    pages={050504},
    year={2017},
    doi={10.1103/PhysRevLett.119.050504}
}

@article{galli2017lovasz,
    title={On the Lov\'{a}sz theta function and some variants},
    author={Laura Galli and Adam N. Letchford},
    journal={Discrete Optimization},
    volume={25},
    pages={159-174},
    year={2017},
    doi={10.1016/j.disopt.2017.04.001}
}

@article{yalovetzky2026quantum-informed,
    title={Quantum-informed portfolio selection: an end-to-end pipeline validated on trapped-ion hardware with real-market data},
    author={Romina Yalovetzky and Martin J.A. Schuetz and Zichang He and Jiayu Shen and Rudy Raymond and Shauna Sahay and Kishore Perla and Ruben S. Andrist and Grant Salton and Helmut G. Katzgraber and Roger Bongiovanni and Niraj Kumar and Rob Otter},
    journal={arXiv:2607.01037},
    year={2026},
    doi={10.48550/arXiv.2607.01037}
}

@book{follmer2016stochastic,
	Author = {H. F\"ollmer and A. Schied},
	Publisher = {Walter de Gruyter},
	Title = {Stochastic Finance: An Introduction in Discrete Time},
	City = {Berlin},
	Year = {2016}
}

@article{ding2025coordinating,
    title={Coordinating decisions via quantum telepathy},
    author={Dawei Ding and Liang Jiang},
    journal={arXiv:2407.21723},
    year={2025},
    doi={10.48550/arXiv.2407.21723}
}

@article{shor1996fault-tolerant,
    title={Fault-tolerant quantum computation},
    author={Peter W. Shor},
    journal={arXiv:quant-ph/9605011},
    year={1996},
    doi={10.48550/arXiv.quant-ph/9605011}
}

@book{feller1991introduction,
    title={An Introduction to Probability Theory and Its Applications},
    author={William Feller},
    edition={3},
    publisher={Wiley},
    year={1991},
    url={https://www.wiley.com/en-us/shop/general-introductory-statistics/an-introduction-to-probability-theory-and-its-applications-volume-1-3rd-edition-p-9780471257080}
}

@book{oksendal2003stochastic,
    title={Stochastic Differential Equations},
    author={Bernt {\O}ksendal},
    publisher={Springer Berlin, Heidelberg},
    edition={6},
    year={2003},
    doi={10.1007/978-3-642-14394-6}
}

@article{dhaene1999safest,
    title={The safest dependence structure among risks},
    author={Jan Dhaene and Michel Denuit},
    journal={Insurance: Mathematics and Economics},
    volume={25},
    issue={1},
    pages={11--21},
    year={1999},
    doi={10.1016/S0167-6687(99)00009-8}
}

@article{bernstein2017post-quantum,
    title={Post-quantum cryptography},
    author={Daniel J. Bernstein and Tanja Lange},
    journal={Nature},
    volume={549},
    pages={188--194},
    year={2017},
    doi={10.1038/nature23461}
}

@article{vertesi2008more,
    title={More efficient Bell inequalities for Werner states},
    author={T. V\'{e}rtesi},
    journal={Physical Review A},
    volume={78},
    pages={032112},
    year={2008},
    doi={10.1103/PhysRevA.78.032112}
}
\end{document}